\documentclass[aps,showpacs,amssymb,amsfonts,superscriptaddress,pra,twocolumn,nofootinbib,10pt]{revtex4-1}
\usepackage{bm,bbm}
\usepackage{times}
\usepackage{rotating}
\usepackage{graphicx,amsthm,epstopdf,amsmath}
\usepackage{mwe}
\usepackage{subfigure,color}
\usepackage[draft]{fixme}
\usepackage{changes,color,colortbl}
\usepackage[super]{nth}
\usepackage{multirow}
\usepackage[top=1.25in, bottom=1.25in, left=0.70in, right=0.70in]{geometry}
\definecolor{myurlcolor}{rgb}{0,0,0.7}
\definecolor{myrefcolor}{rgb}{0.8,0,0}
\usepackage{hyperref}
\hypersetup{colorlinks, linkcolor=myrefcolor,
citecolor=myurlcolor, urlcolor=myurlcolor}
\usepackage{hyperref}
\usepackage{mathtools}
\usepackage{enumerate}
\usepackage{booktabs}

\DeclarePairedDelimiter\floor{\lfloor}{\rfloor}
\newcommand{\beq}[0]{\begin{equation}}
\newcommand{\eeq}[0]{\end{equation}}
\newcommand{\bw}[0]{\begin{widetext}}
\newcommand{\ew}[0]{\end{widetext}}
\newcommand{\bc}[0]{\begin{center}}
\newcommand{\ec}[0]{\end{center}}
\newcommand{\bwn}[0]{\begin{widetext}\begin{eqnarray}}
\newcommand{\ewn}[0]{\end{eqnarray}\end{widetext}}
\newcommand{\beqn}[0]{\begin{eqnarray}}
\newcommand{\eeqn}[0]{\end{eqnarray}}

\newcommand{\uroj}[0]{\mathrm{i}}
\newcommand{\eksp}[0]{\mathrm{e}}
\newcommand{\proj}[1]{|#1\rangle \langle #1|}
\newcommand{\ket}[1]{|#1\rangle}
\newcommand{\bra}[1]{\langle #1 |}

\newcommand{\outerp}[2]{\ket{#1}\! \bra{#2}}

\newcommand{\tr}[0]{\mathrm{tr}}

\newcommand{\jedynka}[0]{\mathbbm{1}}

\newcommand{\non}[0]{\nonumber\\}

\def\calA{{\cal A}}

\def\calG{{\cal G}}
\def\calH{{\cal H}}

\def\calP{{\cal P}}
\def\calQ{{\cal Q}}
\def\calR{{\cal R}}
\def\calS{{\cal S}}

\def\frakF{{\frak F}}

\def\frakS{{\frak S}}

\def\frakX{{\frak X}}

\newtheorem{thm}{Theorem}

\newtheorem{lem}[thm]{Lemma}
\newtheorem{defi}[thm]{Definition}

\newcommand{\beu}{\begin{equation}}
\newcommand{\eeu}{\end{equation}}

\newcommand{\be}{\begin{eqnarray}}
\newcommand{\ee}{\end{eqnarray}}

\newcommand{\ba}{\begin{array}}
\newcommand{\ea}{\end{array}}

\newcommand{\cee}[1]{\mathbb{C}^{#1}}

\newcommand{\tylda}{\tilde{d}}
\newcommand{\ii}[2]{\textbf{i}_{#1,#2}}
\newcommand{\iii}[2]{\textbf{i}_{#1,#2}\boxplus \textbf{1}}
\newcommand{\aaa}[2]{\textbf{A}_{#1,#2}}

\definecolor{Gray}{gray}{0.9}
\begin{document}
\title{Entanglement of genuinely entangled subspaces and states:\\ exact, approximate, and numerical results }
\author{Maciej Demianowicz}
\affiliation{{\it\small Atomic and Optical Physics Division, Department of Atomic, Molecular and Optical Physics, Faculty of Applied Physics and Mathematics,
Gda\'nsk University of Technology, Narutowicza 11/12, 80–-233 Gda\'nsk, Poland}}
\author{Remigiusz Augusiak}
\affiliation{{\it\small Center for Theoretical Physics, Polish Academy of Sciences, Aleja Lotnik\'ow 32/46, 02-668 Warsaw, Poland}}

\begin{abstract}
Genuinely entangled subspaces (GESs) are those subspaces of multipartite Hilbert spaces that consist  only of  genuinely multiparty entangled  pure states. They are natural generalizations of the well-known notion of completely entangled subspaces , which by definition are void of fully product vectors. Entangled subspaces are an important tool of quantum information theory as they directly lead to constructions of entangled states, since any state supported on such a subspace is automatically entangled. Moreover, they have also proven useful  in the area of quantum error correction.
In our recent contribution [M. Demianowicz and R. Augusiak, Phys. Rev. A \textbf{98}, 012313 (2018)], we have 
 studied  the notion of a GES qualitatively in relation to so--called nonorthogonal unextendible product bases
 and provided a few  constructions of such subspaces.
 The main aim of the present work is to perform a quantitative study of the entanglement properties of GESs. First, we show how one can attempt to compute analytically the subspace entanglement, defined as the entanglement of the least entangled vector from the subspace, of a GES and
illustrate our method by applying it to a new class of GESs. 
 Second, we show that certain semidefinite programming relaxations can be exploited to estimate the entanglement of a GES and apply this observation to a few classes of GESs revealing that in many cases the method provides the exact results. Finally, we study the entanglement of certain states supported on GESs, which is compared to the obtained values of the entanglement of the corresponding subspaces, and find the white--noise robustness of several GESs. In our study we use the (generalized) geometric measure as the quantifier of entanglement.
\end{abstract}
\maketitle
\section{Introduction} 

Genuinely multiparty entangled (GME) states, that is states not displaying any form of separability and as such representing the strongest form of entanglement in many body systems, have become an important resource in many information processing protocols over the recent years (see, e.g., \cite{toth-metro,grover-gme,Epping-qkd}).
Due to their significance, there has been a tremendous amount of  research in the literature aimed at understanding their properties (see, e.g.,  \cite{Osterloh-2014,goyeneche-karel,polacos-gme-local,BrunnerExp2017}). While there  has been a lot of progress in the area, still many facets of entanglement in systems of many particles have remained unexplored or less studied.
In an effort to contribute to this line of research, we have recently proposed to analyze in more detail subspaces that only consist of GME states; we called them genuinely entangled subspaces (GESs) \cite{upb-to-ges} (see also \cite{schmidt-rank}). They are the natural analogs of the well--studied completely entangled subspaces (CESs), which are void of fully product vectors \cite{ces-partha,ces-bhat}. Entangled subspaces comprise a particularly important tool of quantum information theory as they allow for general constructions of entangled states, since any state supported on such a subspace is necessarily entangled. Importantly, in the case of GESs such constructed states are GME. Furthermore, particular classes of  entangled subspaces --- perfectly entangled, or $k$--totally entangled, $k$--uniform, ones  \cite{ces-partha} --- have found an application in quantum error correction (QEC) \cite{GourWallach,zahra,ball,felix-arxiv,AME-alsina}. The case $k=\floor{n/2}$ corresponds to certain types of GESs and it is directly related to the notion of absolutely maximally entangled (AME) states  (see, e.g., \cite{Helwig,dardo-AME}).

The attempt at a characterization of GESs made by us in \cite{upb-to-ges} was qualitative, in the sense that we have only considered the problem of their general constructions in setups with an arbitrary number of parties holding subsystems of arbitrary local dimensions (see \cite{Wang2019-ges} for recent advances). This has been linked with the notion of the unextendible product bases \cite{upb-bennett}, another very powerful tool with diverse applications (see, e.g., \cite{DeRinaldis,Duan2010,upb-prl}). Clearly, however, the quantitative description of GESs (or, more generally, any subspaces) is also vital, as it provides a means of comparing them and potentially deciding on their usefulness in certain tasks, in particular those where the amount of entanglement is the figure of merit. So far, this problem has not been considered in the literature (albeit see \cite{ManikBanik-ges}, where the distillability across bipartite cuts has been investigated) and the present paper aims at filling this gap.

There are two main approaches to the problem of quantifying entanglement of a subspace \cite{average-ent,average-Gconcurrence,GourWallach} (see, e.g., \cite{Gour-max-ent} for other ways). In the first one, one asks about the average entanglement over all pure states in a subspace \cite{average-ent,average-Gconcurrence}, in the other, the question is how much entangled is the least--entangled vector in a subspace (see, e.g., \cite{GourWallach}). While both appear equally significant, it is the second one we pursue in the present paper as our main method. From the practical point of view, this approach is relevant for scenarios, in which entangled states drawn from a subspace are a resource and one is interested in the estimation of the worst--case performance of the protocol. In our study we use the geometric measure (GM) of entanglement \cite{Shimony-geometric,BL-geometric,WeiGoldbart}  and its variant, the generalized geometric measure (GGM) \cite{GGM-multiparty}, suitable for GME detection. This choice is motivated by their usefulness in different areas of quantum information theory (see, e.g., \cite{Markham-2007,orus-2008,hayashi-gm-application,gross-to-entangled}).
We present two general methods of computation of the entanglement of a subspace as measured by the geometric measures and show their applicability on a new class of $N$--partite GESs in a $\cee{2}\otimes (\cee{d})^{\otimes (N-1)}$ setup. In particular, we find analytically the GGM of these subspaces for arbitrary $N$ and $d$. Further, the choice of  measures allows us to lower bound the subspace entanglement using semidefinite programming (SDP) relaxations, which in many cases turn out to provide the exact results.
In addition, we also consider other approaches to the problem and analyze
the entanglement properties of states, which are normalized projections on GESs, and investigate
 the white--noise tolerance of such states. These two additional quantifiers, although not standard, are expected to convey some supplementary information about the entanglement of a subspace.
 In this part of our research, we again use the SDP relaxations but also some other tools such as entanglement witnesses, the PPT mixtures, and a direct numerical algorithm for approximating the geometric measures.

The paper is organized as follows. In the preliminary section, Sec. \ref{pre}, we introduce the notation and the terminology. Then, in Sec. \ref{splatanie-gesy}, we recall the definition of the entanglement of a subspace and present two methods of its computation. In Sec. \ref{ent-new-ges}, we apply these methods to find the entanglement of a new class of GESs. Section \ref{sdp-boundy} puts forward SDP bounds on the entanglement of a subspace and investigates their performance for a few classes of GESs. In Sec. \ref{different-methods}, we turn our attention to other methods of quantifying the subspace entanglement, namely, the entanglement of normalized projections on GESs and the white--noise tolerances of such states, and consider several methods of their computation.
We conclude in Sec. \ref{konkluzje}, where we also state some open questions and propose future research directions.

\section{Preliminaries}\label{pre}

In this section we briefly introduce the necessary terminology and the notation.

\paragraph{Notation.} In the paper we deal with finite--dimensional $N$-partite product Hilbert spaces
$\mathbbm{C}^{d_1}\otimes\ldots\otimes \mathbbm{C}^{d_N}$, 
with $d_i$ standing for the dimension of the local Hilbert space of system $A_i$;  the shorthand 
$\textbf{A}:= A_1A_2\dots A_N $ denotes all subsystems. Pure states are denoted as $\ket{\psi},\ket{\varphi},\cdots$, with subscripts corresponding to respective subspaces if necessary, e.g., $\ket{\psi}_{A_1A_2\dots}$. The same convention applies to mixed states, that is we write, e.g., $\rho_{A_1A_2}$ for a state with subsystems held by $A_1$ and $A_2$. For few subsystems the denotations $A$, $B$, \dots, will be used. The standard notation for tensor products of basis vectors is employed: $\ket{ij}=\ket{i,j}:=\ket{i}\otimes \ket{j}$.

\paragraph{Entanglement.} An $N$--partite pure state $\ket{\psi}_{A_1\dots A_N}$ is called {\it fully product} if it is possible to write it as 
\begin{equation}
\ket{\psi}_{A_1\cdots A_N}=\ket{\varphi}_{A_1}\otimes\cdots \otimes\ket{\xi}_{A_N}.
\end{equation}
Otherwise it is said to be {\it entangled}. An important class of such states is  genuinely multiparty entangled  ones. A multipartite pure state is called {\it genuinely multiparty entangled} (GME) if 
\begin{equation}
\ket{\psi}_{A_1\cdots A_N}\ne\ket{\varphi}_{K}\otimes \ket{\phi}_{\bar{K}}
\end{equation}
for any bipartite cut (bipartition) $K | \bar{K}$, where $K$ is a subset of $\textbf{A}$ and $\bar{K}:=\textbf{A}\setminus K$ denotes the rest of the parties.
A paradigmatic example of such a state is the $N$-qubit Greenberger-Horne-Zeilinger (GHZ) state:
%
$\ket{\mathrm{GHZ}_N}=(1/{\sqrt{2}})\left(\ket{0}^{\otimes N}+\ket{1}^{\otimes N}\right)$.
%

States which are not GME, i.e., do admit the form   $\ket{\psi}_{A_1\cdots A_N} = \ket{\varphi}_{K}\otimes \ket{\phi}_{\bar{K}}$, are called {\it biproduct}. It then follows that fully product states are a subclass of the biproduct ones. Let us finally stress that within this terminology a biproduct state is entangled if it is not  fully product.

Generalization of these concepts to the mixed state domain is nontrivial. A mixed state $\rho_{\textbf{A}}$ is said to be {\it fully separable} if it admits the form
\beqn
\rho_{\textbf{A}}=\sum_i p_i \varrho_{A_1}^{(i)}\otimes  \dots \otimes \gamma_{A_N}^{(i)}.
\eeqn
A state which is not fully separable is {\it entangled}.
An entangled multipartite mixed state is called {\it genuinely multiparty entangled} (GME) if
\beqn
\rho_{\textbf{A}} \ne \sum_{K|\bar{K}} p_{K|\bar{K}} \sum_i   q_{K|\bar{K}}^{(i)} \varrho^{(i)}_{K} \otimes \sigma^{(i)}_{\bar{K}},
\eeqn
where the first sum goes over all  bipartitions of $\textbf{A}$.
If a state does admit the decomposition as above, it is called
{\it biseparable}. Just as previously in the case of pure states, we emphasize that biseparable states may be entangled.

\paragraph{Genuinely entangled subspaces.} There exist subspaces composed solely of entangled pure states; they are called completely entangled subspaces (CESs) \cite{ces-partha,ces-bhat}.
 This notion is naturally generalized to the case of GME. Formally, a  subspace $\calG$ of a multipartite Hilbert space is called a genuinely entangled subspace (GES) if all $\ket{\psi}\in \calG$ are genuinely entangled \cite{upb-to-ges} (see also \cite{schmidt-rank,ces-partha}). A simple example of a two-dimensional GES is the subspace spanned by the $W$ state, $\ket{W}=1/\sqrt{N}(\ket{00\dots 001}+\ket{00\dots 010}+\dots \ket{10\dots 000})$, and its complement $\tilde{W}$,  $\ket{\bar{W}}=\sigma_x^{\otimes N} \ket{W}$ \cite{WW-GES}. 
%
%
A few general constructions of higher dimensional GESs have been recently given in \cite{upb-to-ges}, where the notion has been linked to the notion of the unextendible product bases. In fact, the subspaces constructed there are our test--ground cases in the present paper. They are introduced in further parts of the present paper. 

An important observation regarding CESs and GESs is that mixed states supported on them are, respectively, entangled and genuinely multiparty entangled. As such they provide an important tool to construct (genuinely) entangled mixed states.

\section{Entanglement of genuinely entangled subspaces: definition and methods of computation}\label{splatanie-gesy}

Following  \cite{GourWallach}, we define the {\it entanglement of a subspace} $\calS$ (or the {\it subspace entanglement} of $\calS$), as measured by $E$, through
\beqn\label{min-subspace}
E(\calS)=\min_{\ket{\psi}\in \calS} E(\ket{\psi}),
\eeqn
where $E$ is a measure of multipartite entanglement.
Importantly, if this measure is chosen to be non--zero exclusively on GME states, $E(\calS)$ will be a quantifier of genuine entanglement of a subspace.

We pick the geometric measure (GM)  and the generalized geometric measure (GGM) of entanglement as the quantifiers $E$ in our further considerations.  For pure states the GM is defined as \cite{BL-geometric,Shimony-geometric}
\begin{equation}
E_{GM}(\ket{\psi})=1-\max_{\ket{\psi_{\mathrm{prod}}}}|\langle\psi_{\mathrm{prod}}|\psi\rangle|^2
\end{equation}
with the maximization  performed over fully product vectors.
%
Only a slight modification is needed to make this measure quantify solely genuine multipartite entanglement. Namely, one defines the generalized geometric measure  (GGM) of entanglement \cite{Senowie-GGMoE}:
\begin{equation}
E_{GGM}(\ket{\psi})=1-\max_{\ket{\psi_{\mathrm{biprod}}}}|\langle\psi_{\mathrm{biprod}}|\psi\rangle|^2
\end{equation}
with the maximization this time over all pure states that are biproduct. It is obvious that it serves the purpose. The GGM has been shown to be analytically computable for pure states \cite{Senowie-GGMoE}.

For a mixed sate $\rho$, the (G)GM is defined through the standard convex roof construction, that is 
\beqn
E_{(G)GM}(\rho)=\min_{\{p_i,\ket{\psi_i}\}} \sum_i E_{(G)GM}(\ket{\psi_i}),
\eeqn
where the minimum is computed over all pure state ensembles of the state, i.e., $\rho=\sum_i p_i \proj{\psi_i}$.

Clearly, for any subspace $\calS$ it holds that
\beqn \label{state-vs-subspace}
E_{(G)GM} (\rho) \ge E_{(G)GM} (\calS), \quad \mathrm{supp} (\rho) \subseteq \calS,
\eeqn
where $\mathrm{supp} (\rho)$ is the support of the density matrix $\rho$. This property actually holds for any entanglement measure extended from pure states to the mixed--state domain by the convex roof.

Following \cite{Branciard}, we rewrite the right--hand side of (\ref{min-subspace}) for the present choice of the measure to a form useful for further treatment:

\begin{eqnarray}\label{przeksztalcenie}
E_{(G)GM}(\calS) &\equiv& 
\min_{\ket{\psi}\in \calS}E_{(G)GM}(\ket{\psi})\non &=&\min_{\ket{\psi}\in \calS}\left(1-\max_{\ket{\psi_{\mathrm{(bi)prod}}}}|\langle\psi_{\mathrm{(bi)prod}}|\psi\rangle|^2\right)\nonumber\\&=&1-\max_{\ket{\psi_{\mathrm{(bi)prod}}}}\max_{\ket{\psi}\in \calS}|\langle\psi_{\mathrm{(bi)prod}}|\psi\rangle|^2\nonumber\\
&=&1-\max_{\ket{\psi_{\mathrm{(bi)prod}}}}\langle\psi_{\mathrm{(bi)prod}}|\calP_{\calS}|\psi_{\mathrm{(bi)prod}}\rangle\nonumber\\
&=&\min_{\ket{\psi_{\mathrm{(bi)prod}}}}\langle\psi_{\mathrm{(bi)prod}}|\calP_{\calS}^{\perp}|\psi_{\mathrm{(bi)prod}}\rangle,
\end{eqnarray}
where $\calP_{\calS}$ projects onto $\calS$
and $\calP_{\calS}^{\perp}$ -- onto $\calS^{\perp}$, that is $\calP_{\calS}+\calP_{\calS}^{\perp}=\jedynka$. The crucial transition from the third line to the fourth follows from the fact that for a given $\psi_{\mathrm{prod}}$, the vector maximizing the quantity will be the (normalized) projection of $\psi_{\mathrm{prod}}$ onto $\calS$. The great value of this reformulation lies in the fact that we now  have only one optimization to perform. 

Now, if $\calS$ contains (bi)product vectors then this quantity will simply  give  zero. However, in the opposite case, it is certainly nonzero, signifying (genuine) entanglement of $\calS$. 

Notice that the GGM of a subspace can also be written as:
\beq\label{ggm-gm-S}
E_{GGM}(\calS)=\min_{K|\bar{K}} E_{GM}^{K|\bar{K}}(\calS),
\eeq
where $E_{GM}^{K|\bar{K}}(\calS)$ denotes the GM of the subspace across a particular bipartition and the minimization is over all bipartitions.
 This emphasizes the fact that although we deal here with genuine entanglement the problem reduces to a repeated analaysis of a bipartite case. This is the great feature of the GGM making it computable in many cases.
 We will use this formulation in one of the proofs. We must stress, however, that this reasoning only applies to subspaces, not states.

In passing, we note that one could also consider "intermediate'' geometric measures, where instead of considering $N$-- product (fully product) or $2$--product (biproduct) vectors, the maximization is performed over $k$-- product vectors. This has been considered, e.g., in \cite{Blasone2008}.

\subsection{Method of projecting onto a subsystem} \label{projecting-method}

Here we describe a general method of computing the entanglement of a subspace based on the observation (\ref{przeksztalcenie}). We will refer to this method as the {\it method of projecting onto a subsystem}.

Let us first consider the case of the GGM and rewrite $E_{GGM}(\calS)$ more explicitly as:
\beqn
\hspace{-0.8cm}E_{GGM}(\calS)&&\non
&&\hspace{-1.cm}=\min _{K|\bar{K}} \min_{ \ket{ \varphi_K } \otimes \ket{ \bar{ \varphi }_{ \bar{K} } } } \bra{\varphi_K}\otimes \bra{\bar{\varphi}_{\bar{K}} } \calP_{\calS}^{\perp} \ket{\varphi_K}\otimes \ket{\bar{\varphi}_{\bar{K}}} \label{min}\\
&&\hspace{-1cm}= 1- \max _{K|\bar{K}} \max_{ \ket{ \varphi_K } \otimes \ket{ \bar{ \varphi }_{ \bar{K} } } } \bra{\varphi_K}\otimes \bra{\bar{\varphi}_{\bar{K}} } \calP_{\calS} \ket{\varphi_K}\otimes \ket{\bar{\varphi}_{\bar{K}}}, \label{max}
\eeqn
where $K|\bar{K}$ denotes a bipartition of the parties and $\ket{\varphi_K}$ and  $\ket{ \bar{ \varphi }_{ \bar{K} } }$ are vectors on $K$ and $\bar{K}$, respectively. We have reversed the order of representations of $E(\calS)$ in comparison to (\ref{przeksztalcenie}) because, in fact, the second form, i.e., (\ref{max}), will be more useful to us. This is due to the fact that we will more  easily write out a basis for $\calS$ than for $\calR$.  We note, however, that the method works equally well if one uses the representation (\ref{min}).

Defining the matrix
\beqn \label{matrix-S}
\hspace{-0.7cm}
\frakS_{\bar{K}}:= \left(\jedynka_K\otimes \bra{\bar{\varphi}_{\bar{K}} }\right)\calP_{\calS}\left(\jedynka_K\otimes \ket{\bar{\varphi}_{\bar{K}}}\right) \equiv \bra{\bar{\varphi}_{\bar{K}} }\calP_{\calS}\ket{\bar{\varphi}_{\bar{K}}},
\eeqn
we can rewrite Eq. (\ref{max}) as
\beqn\label{ggm-lambda}
E_{GGM}(\calS)=1- \max_{K|\bar{K}}\lambda_{\max}(\frakS_{\bar{K}}),
\eeqn
where $\lambda_{\max}$ is the largest eigenvalue maximized over all choices of $\tilde{\varphi}_{\bar{K}}$ of the matrix $\frakS_{\bar{K}}$ for the cut $K|\bar{K}$. 
The problem of the computation of the entanglement of a subspace has thus been reduced to the problem of the computation of the maximal eigenvalues of certain matrices and then picking the largest among them.

Clearly, one could choose the subsystem $K$ to perform the projection on in (\ref{matrix-S}). Which subsystem we choose in practice is dictated by the simplicity of the resulting computations of the largest eigenvalues.

In the case of the GM in which the optimization is over fully product vectors one defines the counterpart of the matrix $\frakS_{\bar{K}}$ as:
\beqn\label{S-matrix-A}
\frakS_{A_i}:=\left(\bigotimes_{\substack{k=1\\ i\ne k}}^N\bra{\varphi_k}\right) \calP_{\calS} \left(\bigotimes_{\substack{k=1\\ i\ne k}}^N\ket{\varphi_k} \right),
\eeqn
where $\ket{\varphi_k}\in \calH_{A_k}$. In analogy to the case above, all $\frakS_{A_i}$'s must be considered and the maximum of the set of the largest eigenvalues of all such matrices must be found.

If the vectors spanning $\calS$ share some nice structural properties, the largest eigenvalues of $\frakS$'s can be found analytically. Otherwise, we need to resort to numerical calculations. We consider both situations in what follows.

\subsection{Seesaw iteration}\label{see-saw-method}

An alternative approach to the optimization problems above is a {\it seesaw} iteration, which is as follows. We start with an initial vector (the subscripts enumerate the step of the algorithm, the superscripts the number of the party): $\ket{\psi_0^{(1)}}\otimes \ket{\psi_0^{(2)}} \otimes \cdots \otimes \ket{\psi_0^{(N)}}$. This vector can be chosen at random. We then choose some small number $\varepsilon >0$ and construct the following matrix
\begin{equation}
S_{00\dots 0}^{(1)}:= \bra{\psi_0^{(2)}} \cdots\bra{\psi_0^{(N)}} \calP_{\calS} \ket{\psi_0^{(2)}} \cdots\ket{\psi_0^{(N)}}; 
\end{equation}
the eigenvector corresponding to its largest eigenvalue is set to be $\ket{\psi_1^{(1)}}$. Then the matrix  
\begin{equation}
S_{10\dots 0}^{(2)}:= \bra{\psi_1^{(1)}}\bra{\psi_0^{(3)}} \cdots\bra{\psi_0^{(N)}} \calP_{\calS}  \ket{\psi_1^{(1)}}\ket{\psi_0^{(3)}} \cdots\ket{\psi_0^{(N)}}
\end{equation}
is constructed; the eigenvector corresponding to its largest eigenvalue is set as $\ket{\psi_1^{(2)}}$. The procedure is repeated for all the parties to get the first approximation of the optimal product vector $\ket{\psi_1^{(1)}}\otimes \ket{\psi_1^{(2)}} \otimes \cdots \otimes \ket{\psi_1^{(N)}}$, which ends the first step of the algorithm. The output is accepted if 
\begin{eqnarray}
\hspace{-3cm}&&\bra{\psi_1^{(1)}} \cdots\bra{\psi_1^{(N)}} \calP_{\calS} \ket{\psi_1^{(1)}} \cdots\ket{\psi_1^{(N)}}\nonumber\\
&&-\bra{\psi_0^{(1)}} \cdots\bra{\psi_0^{(N)}} \calP_{\calS} \ket{\psi_0^{(1)}} \cdots\ket{\psi_0^{(N)}}< \varepsilon.
\end{eqnarray}
Otherwise, the next step is performed. The procedure is repeatedly used until the required precision $\varepsilon$ is reached. 
The algorithm needs to be run for a number of initial states to increase the chance of avoiding a local maximum. 

For subspaces with a nice structure one can expect to be able to perform some number of steps analytically in the iteration above.
The advantage of the seesaw approach in comparison to a direct optimization over parameters is its simplicity and the speed of execution.

Notice that the method of projecting onto a subsystem exposed in previous subsection can be seen as  a one-shot analytical seesaw method.

\section{Entanglement of a class of genuinely entangled subspaces of  $\cee{2}\otimes (\cee{d})^{\otimes (N-1)}$} \label{ent-new-ges}

In this section we apply our methods to a new class of $N$--partite genuinely entangled subspaces with one qubit subsystem and the rest $d$--level systems, i.e., a subspace of  $\cee{2}\otimes  (\cee{d})^{\otimes (N-1)}$.
For $N=2$ the subspace reduces to a completely entangled one and we begin with this basic case.

\subsection{A completely entangled subspace of $\cee{2}\otimes \cee{d}$} \label{ces-new}

The method of projecting onto a subsystem put forward above in Sec. \ref{projecting-method} is clearly applicable in this case with the simplification being that we only  deal  with the GM due the bipartite nature of the problem.

Let us introduce the relevant subspace. 

\begin{defi}\label{ces-def}
The subspace $\calS_{2\times d}^{\theta} \subset \cee{2}\otimes \cee{d}$ is given by the span of the following vectors:
\beqn\label{ces}
\ket{\phi_i}_{AB}=a \ket{0}_A\ket{\psi_i}_B+b\ket{1}_A\ket{\psi_{i+1}}_B,
\eeqn
$  i=0,1,\dots, d-2$,	with $a=\cos (\theta/2)$ and $b=\eksp^{\uroj \xi} \sin (\theta/2)$, $\theta \in (0,\pi)$, $\xi \in [ 0, 2\pi)$, and $\bra{\psi_i}\psi_j\rangle=\delta_{ij}$.
\end{defi}
 Clearly, $\dim \calS_{2\times d}^{\theta}=d-1$, which is also the maximal available dimension of a CES in this scenario \cite{ces-bhat,ces-partha}.

 The following result giving the entanglement of this subspace will serve as a basis for our further computations.

\begin{thm}\label{ces-ent}
	Let $d \ge 3$.
$\calS_{2\times d}^{\theta}$ is a CES with the subspace entanglement as measured by the GM  given by:
\beq\label{ces-gm}
\displaystyle
E_{GM}(\calS_{2\times d}^{\theta})=\frac{1}{2} \left( 1-
\sqrt{1-\sin^2\theta \sin^2 \left(\frac{ \pi}{d}\right) }   \right).
\eeq
In particular, for $a=b=1/\sqrt{2}$, i.e., $\theta=\pi/2$, the entanglement is
\beq
\displaystyle
E_{GM}\left(\calS_{2\times d}^{\pi/2}\right)=\frac{1}{2}\left(1-\cos \frac{\pi}{d}\right)=\sin^2  \frac{\pi}{2d}.
\eeq
\end{thm}
\noindent
\begin{proof}
We omit the detailed proof here, which has been moved to Appendix \ref{app-proof-1} and only present the main ingredient.
The relevant matrix (\ref{S-matrix-A}) has the tridiagonal form:
\beqn \label{tridiagonal-x}
 \left( \begin{array}{cccccc}
	\alpha&g &0 & \cdots &0 &0\\
	g^* & 	\alpha+\beta&g& \cdots &0&0\\
	0& g^* &\alpha+\beta & \cdots &0&0\\
	\vdots & \vdots & \ddots & \ddots &\vdots & \vdots\\
	0& 0  &\ddots &\alpha+\beta&g &0\\
	0& 0 &\cdots & g^*  &\alpha+\beta &g\\
	0&0 &\cdots & 0&g^*  &\beta
\end{array}     \right),
\eeqn
with $\alpha= |a x_0|^2$ , $\beta= |b x_1|^2$, $g=  ax_1 (bx_0)^* $.
One finds its eigenvalues to be \cite{losonczi}:
\beqn \label{w-wlasne}
\lambda_k \left(x_0\right)&=& \alpha+\beta+2|g| \cos \frac{k \pi}{d}\\
&=&  |a x_0|^2 + |b x_1|^2 + 2 |x_0x_1ab| \cos \frac{k \pi}{d} \nonumber
\eeqn
for $k=1,2,\dots, d-1$ and $\lambda_d=0$. 	
The task is to find 
\beqn
\lambda_{\mathrm{max}}:=\max_{x_0,k} \lambda_k ( x_0)
\eeqn
and the optimization results in  (\ref{ces-gm}). \quad \quad\quad\quad\quad\quad\quad\quad\quad 
\end{proof}

Not surprisingly, the subspace entanglement is a decreasing function of $d$ with the maximum at $a=1/\sqrt{2}$ ($\theta=\pi/2$) for any $d$.  When $d=2$ (not covered by the theorem) we only have one vector, whose entanglement is immediately found to be $\min \{\sin^2 (\theta/2),\cos^2 (\theta/2)\}$. This is also the GM of any of the spanning vectors in (\ref{ces}). In Fig. \ref{ent-subspace-a} we plot $E_{GM}(\calS_{2\times d}^{\theta})$ as a function of $\theta$ for several values of the dimension $d$. 

\begin{figure}[h!]
	\includegraphics[height=5cm,width=8cm]{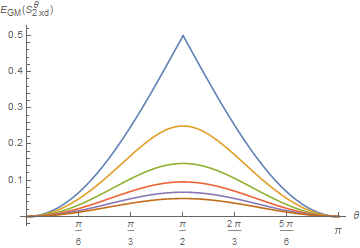}
	\caption{ $E_{GM}(\calS_{2\times d}^{\theta})$ as a function of $\theta$ for $d=3$ (second from the top curve), $4,5,6,7$ (lower most curve). For reference we also put the entanglement of the vector in the case of $d=2$, or equivalently, any vector in (\ref{ces}) (upper most curve).
	}\label{ent-subspace-a}
\end{figure}

\subsection{$N$ party GES $\calS_{2\times d^{N-1}}^{\theta}$: the case of the GGM}\label{th1-generalized}

Let us now move to the multiparty case. In this section we introduce a multipartite subspace, which is a generalization of the CES from Definition \ref{ces-def} to $N\ge 3$ parties, and compute its GGM.

\begin{defi}\label{ges-def}
The subspace $\calS_{2\times d^{N-1}}^{\theta} \subset \cee{2}\otimes (\cee{d})^{\otimes (N-1)}$ is given by the span of the following vectors:
\beqn\label{ges}
&&\ket{\Phi_{i_2 \dots i_N}} _{\textbf{A}}=\\
&&a \ket{0}_{A_1}\left( \bigotimes_{m=2}^N \ket{\psi_{i_m}^{(m)}}_{A_m}\right)+b\ket{1}_{A_1}\left( \bigotimes_{m=2}^N \ket{\psi_{i_m+1}^{(m)}}_{A_m}\right),\nonumber
\eeqn
$i_m=0, 1,\dots, d-2$, with $a=\cos (\theta/2)$, $b=\eksp^{\uroj \xi} \sin (\theta/2)$, $\theta \in (0,\pi)$, $\xi \in [0, 2\pi)$, and $\bra{\psi_{i_m}}\psi_{j_m}\rangle=\delta_{i_mj_m}$ for each $m$.
\end{defi}
One easily sees that $\dim \calS_{2\times d^{N-1}}^{\theta} = (d-1)^{N-1}$.  Note that the maximal dimension of a GES in this setup is $d^{N-1}-1$ \cite{upb-to-ges} and it would thus be interesting to see how these subspaces could be completed into  maximal GESs (but not with random vectors, which can always be done).

Using the results of the previous subsection, we now prove the main result of this part of the paper stating that the subspace entanglement of $\calS_{2\times d^{N-1}}^{\theta}$ measured by the GGM is the same as that of the CES $\calS_{2\times d}^{\theta}$ considered above, and the subspace is equally entangled across any cut with respect to this measure. Generally speaking, this equality is due to the fact that the relevant matrices whose eigenvalues need to be maximized bear the same structures as a result of the generalization.

 Clearly, the choice of $\psi$'s does not matter for the entanglement properties of the subspace as different choices for them are related through local unitary operations, which do not change entanglement measures.
  For the clarity of exposition, without loss of generality, we then set $\ket{\psi_{i_m}^{(m)}}_{A_m}=\ket{i_m}_{A_m}$ in the proof. The basis vectors (\ref{ges}) are in this case:
\beqn\label{baza-standardowa-ges}
\ket{\Phi_{i_2\dots i_N}}_{\textbf{A}}&= &a\ket{0}_{A_1}\ket{i_2, \dots, i_N}_{A_2 \dots A_N}\\ && \hspace{+0.3cm}+b\ket{1}_{A_1}\ket{i_2+1, \dots, i_N+1}_{A_2 \dots A_N}. \nonumber
\eeqn

We will then use the following denotations allowing us to keep the formulas cleaner:
\beqn
&& \aaa{k}{l}:=A_k, A_{k+1},\dots, A_l,\\
&&\ii{k}{l}:= i_k,i_{k+1},\dots, i_l \\
&&\iii{k}{l}:=i_k+1,i_{k+1}+1,\dots, i_l+1,\\
&&\sum_{\ii{k}{l}}:= \sum_{i_k,\dots,i_l=0}^{d-2}
\eeqn
for any $k<l$. In this notation the vectors (\ref{baza-standardowa-ges}) can be compactly written as:
\beqn\label{compact-basis}
\hspace{-0.4cm}\ket{\Phi_{\ii{2}{N}}}_{\textbf{A}}= a\ket{0}_{A_1}\ket{\ii{2}{N}}_{\aaa{2}{N}}+b\ket{1}_{A_1}\ket{\iii{2}{N}}_{\aaa{2}{N}}. 
\eeqn

Before we state the main result, let us give a simple lemma which will be crucial for its proof.

\begin{lem}\label{lema-orto}
	Given are operators $R_k$ with orthogonal supports. Let $R:=\sum_k R_k$. The largest eigenvalue of $R$ is given by $\lambda_{\mathrm{max}}(R)=\max_k \lambda_{\mathrm{max}}(R_k) $.
\end{lem}
\begin{proof}
	This is obvious.
\end{proof}

\begin{thm}\label{ges-ggm}
$\calS_{2\times d^{N-1}}^{\theta}$ is a GES with the  subspace entanglement as measured by the GGM  given  by:
\beqn\label{ggm-of-ges}
\displaystyle
\hspace{-0.7cm} 
E_{GGM}(\calS_{2\times d^{N-1}}^{\theta})=\frac{1}{2} \left( 1-
\sqrt{1-\sin^2\theta \sin^2 \left(\frac{ \pi}{d}\right) }   \right).
\eeqn
Moreover, the entanglement of $\calS_{2\times d^{N-1}}^{\theta}$ is the same across any bipartite cut.
\end{thm}

We will use the observation (\ref{ggm-gm-S}) in the proof, which will be split into two parts regarding different types of cuts: 

(i) $A_1|\aaa{2}{N}$,

 (ii) $k|N-k$ parties, $k>1$.

\noindent\textit{Proof of Theorem \ref{ges-ggm}. Case (i): Entanglement across $A_1|\aaa{2}{N}$}.
We construct the matrix (\ref{matrix-S}), by choosing $\bar{K}=A_1$, i.e., projecting onto the $A_1$ subsystem:
\beqn
\frakS_{A_1}(x_0)=\sum_{\ii{2}{N}}  \proj{\varphi_{\ii{2}{N}}}_{\aaa{2}{N}}
\eeqn
where we have defined:
\beqn\label{zrzutowane-A1}
\ket{\varphi_{\ii{2}{N}}}_{\aaa{2}{N}}
&\equiv&\bra{x}_{A_1}\ket{\Phi_{\ii{2}{N}}}_{\textbf{A}} \\
&=&ax_0^*\ket{\ii{2}{N}}_{\aaa{2}{N}}+bx_1^*\ket{\iii{2}{N}}_{\aaa{2}{N}} \nonumber
\eeqn
with $\ket{x}_{A_1}=x_0 \ket{0}_{A_1}+x_1\ket{1}_{A_1}$, $|x_0|^2+|x_1|^2=1$.

We now distinguish two cases: 

(a) $x_0=0$ or $1$, 

(b) $x_0 \ne 0,1$.

In case (a), the vectors (\ref{zrzutowane-A1}) are orthogonal and the non--zero eigenvalues of $\frakS_{A_1}$ are all equal to $|a|^2=\cos^2(\theta/2)$ (for $x_0=1$) or $|b|^2=\sin^2(\theta/2)$ (for $x_0=0$).

In case (b), these vectors  no longer form an orthogonal set. 
Our strategy now will  be to write the matrix $\frakS_{A_1} (x_0)$ as a sum of operators, call them $R_k$'s, with orthogonal supports and then use Lemma \ref{lema-orto} to find its largest eigenvalue.

With this aim, consider a grouping of the vectors according to the index $i_2$. For each value  $i_2=\tilde{i}_2$, we call such a group $F_{\tilde{i}_2}$. Vectors within each group are orthogonal, while the overlapping vectors necessarily come from neighboring groups, which is easily seen directly from Eq. (\ref{zrzutowane-A1}). Moreover, for a given vector there may  exist only one vector with a non-zero overlap with it and some vectors are orthogonal to all the remaining ones, e.g., this happens for the vector $\ket{\phi_{00\dots 0 \: d-2}}$.
These observations will be exploited in 
the construction of  $R_k$'s, which goes as follows.
To construct $R_0$: (i) choose the first vector from $F_0$, i.e., $\ket{\varphi_{0\dots00}}$, (ii) find the vector from $F_1$ which has a non--zero overlap with $\ket{\varphi_{0\dots00}}$, this will be $\ket{\varphi_{1\dots11}}$, (iii) find the vector from $F_2$ having nonzero overlap with $\ket{\varphi_{1\dots11}}$, this vector is $\ket{\varphi_{2 \dots 22}}$, (iv) repeat the procedure until the group $F_{d-2}$ is reached and the vector $\ket{\varphi_{d-2,\dots, d-2,d-2}}$ is picked from this group. Build 
\beq \label{r0}
R_0 =\sum _{p=0}^{d-2} \proj{\varphi_{p\dots pp}}.
\eeq
 To construct $R_1$: (i) take the second vector of $F_0$, i.e., $\ket{\varphi_{0\dots01}}$, (ii) find the vector from $F_1$ which has a nonzero overlap with $\ket{\varphi_{0\dots01}}$, this will be $\ket{\varphi_{1\dots12}}$, (iii) repeat the procedure until the group is reached with no vector with a non--zero overlap with the one from the previous group. It is easy to see that in this case the last vector to be drawn is $\ket{\varphi_{d-3,\dots, d-3,d-2}}$ from $F_{d-3}$, i.e., the procedure terminates faster as there is no vector from $F_{d-2}$ which is suitable (this vector would have been $\ket{\varphi_{d-2,\dots, d-2,d-1}}$). Build $R_1$ similarly to (\ref{r0}). We repeat the whole procedure for all so far unused vectors (from $F_0$, but also from the following groups) to construct the remaining $R_k$'s.  As noted above, some $R_k$'s will simply be equal  to $\proj{\varphi_{i_2 \dots i_N}}$ for some values of indices $i_m$'s.
 
  This procedure decomposes the matrix $\frakS_{A_1}(x_0)$ as the sum $\frakS_{A_1}(x_0)=\sum_k R_k$,
 where all $R_k$'s have mutually orthogonal supports as desired.

Importantly, the eigenvalues of the constituent operators can be now easily found and the use of Lemma \ref{lema-orto} is straightforward.
This is because either (i) they are $\proj{\varphi}$ or 
(ii) have the same structure of the tridiagonal matrix from (\ref{tridiagonal-x}) but the corresponding matrices are of different sizes (and ranks); this is easily seen if we look at the form of the vectors (\ref{zrzutowane-A1}). 
Among the latter operators, $R_0$ has the largest rank (equal to $d-1$) due to the fact that it has been  constructed from the vectors from all the groups $F_{i_2}$'s and its matrix is $d\times d$. It is the unique such operator.
Since the largest eigenvalue of (\ref{tridiagonal-x}) is increasing with $d$,
we need to maximize over $x_0$ the largest eigenvalue of $R_0$. This has already been done in the proof of Theorem \ref{ces-ent}. In  turn, we have:
\beqn\label{lambda-max-theta}
\lambda_{\mathrm{max}}(\frakS_{A_1})=\frac{1}{2} \left( 1+
\sqrt{1-\sin^2\theta \sin^2 \left(\frac{ \pi}{d}\right) }   \right).
\eeqn
Since this is  larger than both $\cos^2(\theta/2)$ and $\sin^2(\theta/2)$ corresponding to the cases of $x_0=0$ or $1$, respectively, we have that in the cut $A_1|\aaa{2}{N}$ the GM of the subspace equals:
\beqn\label{ggm-cutA1}
E_{GM}^{A_1|\aaa{2}{N}}(\calS_{2\times d^{N-1}}^{\theta})=\frac{1}{2} \left( 1-
\sqrt{1-\sin^2\theta \sin^2 \left(\frac{ \pi}{d}\right) }   \right),\non
\eeqn
which is the same as $E_{GM}(\calS_{2\times d}^{\theta})$. This ends the  part of the proof for the current bipartition.
\hspace{+3.7cm}
	$\blacksquare$

Before we move to the case of other cuts, we notice that the above analysis lets us prove indirectly an interesting result which will be essential in the remainder of the proof of Theorem \ref{ges-ggm}. 

\begin{lem}\label{maksymalizacja-ogolna}
	Let $d \ge 3$ and $M \ge 1$.
Consider the matrix:
%
\beqn\label{X-lematowa}
\frakX=
\sum_{\ii{1}{M}} \left( \begin{array}{cc}
	|a|^2|x_{\ii{1}{M}}|^2 & ab^*x_{\ii{1}{M}}^*x_{\ii{1}{M}\boxplus \textbf{1}}  \\
	a^*bx_{\ii{1}{M}}x_{\ii{1}{M}\boxplus \textbf{1}}^*  & 	|b|^2|x_{\ii{1}{M}\boxplus \textbf{1}}|^2
\end{array}     \right), 
\eeqn
 where $a=\cos (\theta/2)$, $b=\eksp^{\uroj \xi} \sin (\theta/2)$, $\theta \in (0,\pi)$, $\xi \in [ 0, 2\pi)$, and $\sum_{i_1,\dots, i_M =0}^{d-1} |x_{\ii{1}{M}}|^2=1$.

The largest eigenvalue of $\frakX$ maximized over the coefficients $x_{\ii{1}{M}}$'s is given by the formula (\ref{lambda-max-theta}), regardless of the value of $M$, i.e., of the number of the indices $i_p$.
\end{lem}

\noindent {\it Proof of Lemma \ref{maksymalizacja-ogolna}}. 	
To compute the entanglement of $\calS_{2\times d^{N-1}}^{\theta}$ across the bipartition $A_1|\aaa{2}{N}$, instead of the matrix $\frakS_{A_1}(x_0)$ consider the complementary one $\frakS_{\aaa{2}{N}} (\textbf{x})$, obtained by projecting onto 
\beq\ket{x}_{\aaa{2}{N}}=\sum_{j_2,\dots,j_N=0}^{d-1} x_{j_{2}\dots j_N} \ket{j_{2}\dots j_N}_{\aaa{2}{N}};
\eeq%
 $\textbf{x}$ denotes the set of the coefficients of $\ket{x}_{\aaa{2}{N}}$. It is easy to realize that it is a two by two matrix of the form:
\beqn\label{frakS-inna}
\hspace{-0.7cm}
\frakS_{\aaa{2}{N}} (\textbf{x})= \sum_{\ii{2}{N}}\left( \begin{array}{cc}
	|a|^2|x_{\ii{2}{N}}|^2 & ab^*x_{\ii{2}{N}}^*x_{\ii{2}{N}\boxplus \textbf{1}}  \\
	a^*bx_{\ii{2}{N}}x_{\ii{2}{N}\boxplus \textbf{1}}^*  & 	|b|^2|x_{\ii{2}{N}\boxplus \textbf{1}}|^2
\end{array}     \right). 
\eeqn

Since the results obtained for the value of the GM with both $\frakS$'s must be the same, we conclude that the largest eigenvalue of $\frakS_{\aaa{2}{N}}$ maximized over $\textbf{x}$ is given by (\ref{lambda-max-theta}), regardless of the number of parties $N$. This proves the claim as $\frakS_{\aaa{2}{N}} (\textbf{x})$ is of the same structure as $\frakX$ from (\ref{X-lematowa}) of Lemma \ref{maksymalizacja-ogolna}.
\hspace{+1.3cm}
$\square$

With this preparation in hands let us then go back to the proof of Theorem \ref{ges-ggm} and consider the cuts with $k$ vs. $N-k$ parties for $k>1$.
\newline\newline
\noindent{\it Proof of Theorem \ref{ges-ggm} (cont'd). Case (ii): Cuts $k|N-k$ parties, $k>1$}.
Clearly, for any cut  $k|N-k$  with $k>1$ we may consider, without loss of generality, the bipartition  $\aaa{1}{k}|\aaa{k+1}{N}$.
We construct the matrix from (\ref{matrix-S}) by projecting onto the subsystem $\aaa{k+1}{N}$:
\beqn\label{zrzutowane-Ak}
\frakS_{\aaa{k+1}{N}}(\textbf{x})=\sum_{\ii{2}{N}} \proj{\xi_{\ii{2}{k}}^{\ii{k+1}{N}}}_{\aaa{1}{k}},
\eeqn
where
\beqn
\ket{\xi_{\ii{2}{k}}^{\ii{k+1}{N}}}_{\aaa{1}{k}}&\equiv&\bra{x}_{\aaa{k+1}{N}}\ket{\Phi_{\ii{2}{N}}}_{\textbf{A}} \\
&=&a x_{\ii{k+1}{N}}^*\ket{0}_{A_1}\ket{\ii{2}{k}}_{\aaa{2}{k}}\non &&\hspace{+0.4cm}+ bx_{\iii{k+1}{N}}^*\ket{1}_{A_1}\ket{\iii{2}{k}}_{\aaa{2}{k}} \nonumber
\eeqn
with the normalized vector on the $\aaa{k+1}{N}$ subsystem:
\beq
 \ket{x}_{\aaa{k+1}{N}}=\sum_{j_{k+1},\dots,j_N=0}^{d-1} x_{\textbf{j}_{k+1,N}} \ket{\textbf{j}_{k+1,N}}_{\aaa{k+1}{N}},
 \eeq
  and $\textbf{x}$ denoting the set of the coefficients $ x_{\textbf{j}_{k+1,N}}$. 

Rewrite now (\ref{zrzutowane-Ak}) as
\beqn
\frakS_{\aaa{k+1}{N}}(\textbf{x})=\sum_{\ii{2}{k}}  R_{\ii{2}{k}}(\textbf{x})
\eeqn
where
\beqn
R_{\ii{2}{k}}(\textbf{x})= \sum_{\ii{k+1}{N}} \proj{\xi_{\ii{2}{k}}^{\ii{k+1}{N}}}_{\aaa{1}{k}}.
\eeqn
Since $\langle \xi_{\textbf{j}_{2,k }}^{\ii{k+1}{N}}\ket{\xi_{\ii{2}{k}}^{\ii{k+1}{N}}}=0$ whenever $\textbf{j}_{2,k }\ne \ii{2}{k}$, the operators $R$'s have orthogonal supports,
which, by Lemma \ref{lema-orto}, means that
\beqn
\lambda_{\mathrm{max}}\left( \frakS_{\aaa{k+1}{N}}\right)
= \max_{\textbf{x},\ii{2}{k}} \lambda_{\max} \left( R_{\ii{2}{k}}(\textbf{x}) \right).
\eeqn
This maximization can be easily done if we realize that all $R$'s have in fact the same structure of the two-by-two matrix considered in Lemma \ref{maksymalizacja-ogolna}. Since the largest eigenvalue does not depend on the number of indices we conclude that $\lambda_{\mathrm{max}}\left( \frakS_{\aaa{k+1},N}\right)$ is again given by (\ref{lambda-max-theta}), and in turn, by (\ref{ggm-lambda}) applied to the particular bipartition, we obtain:
\beqn
E_{GM}^{\aaa{1}{k}|\aaa{k+1}{N}}(\calS_{2\times d^{N-1}}^{\theta})=\frac{1}{2} \left( 1-
\sqrt{1-\sin^2\theta \sin^2 \left(\frac{ \pi}{d}\right) }   \right).\non
\eeqn
In conjunction with (\ref{ggm-cutA1}) this shows that all cuts are equally entangled and we arrive at the claimed result (\ref{ggm-of-ges}). \hspace{+1.5cm}
$\square$

In Appendix \ref{app-th1} we show that in the case of $N=3$ the subspace $\calS_{2\times d^2}^{\pi/2}$  corresponds to the one given in Theorem 1 of Ref. \cite{upb-to-ges}. In fact, this correspondence was our primary motivation for considering such subspaces.

\subsection{$N$ party GES $\calS_{2\times d^{N-1}}^{\theta}$: the case of the GM}\label{GM-ges}

We now move to the computation of the GM of  $\calS_{2\times d^{N-1}}^{\theta}$. Again, for simplicity we set $\ket{\psi_{i_m}^{(m)}}_{A_m}=\ket{i_m}_{A_m}$.

Take the fully product vectors in the  problem (\ref{przeksztalcenie}) to be:
\beqn
\ket{\psi_{\mathrm{prod}}}=\ket{x^{(1)}}_{A_1}\otimes \dots \otimes \ket{x^{(N)}}_{A_N}
\eeqn
with the normalized local vectors:
\beqn\label{iksy}
&&\ket{x^{(1)}}_{A_1}=x_0^{(1)} \ket{0}_{A_1}+x_1^{(1)} \ket{1}_{A_1},\non
&&\ket{x^{(m)}}_{A_m}=\sum_{n=0}^{d-1} x_{n}^{(m)} \ket{n}_{A_m}, \quad m=2,3,\dots, N.
\eeqn
Inserting this into (\ref{przeksztalcenie}) with $\calP_{\calS}$ taken to be the projection onto $\calS^{\theta}_{2\times d^{N-1}}$ we obtain:
\beqn\label{key-problem}
E_{GM} (\calS^{\theta}_{2\times d^{N-1}})=1-\max_{\textbf{x}} F_N(\textbf{x}),
\eeqn
where $\textbf{x}$ denotes the set of all coefficients of $\psi_{prod}$, and
\beqn
&&F_N(\textbf{x})=\\
&&\hspace{+0.5cm}= \sum_{i_2,\dots,i_N=0}^{d-2} \Big| a x_0^{(1)}x_{i_2}^{(2)}\cdots x_{i_N}^{(N)} + b x_1^{(1)}x_{i_2+1}^{(2)}\cdots x_{i_N+1}^{(N)}    \Big| ^2 .\nonumber
\eeqn

Maximization of this quantity can be approached with the seesaw algorithm (section \ref{see-saw-method}).
Unfortunately, it is not possible to obtain an exact formula through this approach, nevertheless, an easily computable bound can be given. The details are as follows.
For simplicity we consider the case $a=b$ but the essential arguments remain unchanged outside this specialized case.

At the beginning, we set all the coefficient on parties $A_2,\dots, A_N$ equal, i.e., $x_{i_k}^{(k)}=1/\sqrt{d}$, which results in the following quantity to be maximized  :
\beqn\displaystyle
F_1(\textbf{x}_1)= c_1(d,N)  \Big| x_{0}^{(1)} + x_{1}^{(1)}    \Big| ^2, 
\eeqn
$c_1(d,N)=\frac{1}{2}[(d-1)/d]^{N-1}$, $\textbf{x}_1=(x_0^{(1)},x_1^{(1)})$.
Clearly,  the factor in front is not important and the optimal values (up to an irrelevant phase) are $x_0^{(1)}=x_1^{(1)}=1/\sqrt{2}$.

 We then set the obtained values on $A_1$, keeping the coefficients equal on $A_3 \dots A_N$, and the resulting quantity to be maximized in the second step of the first iteration is given by :
\beqn\displaystyle
F_2(\textbf{x}_2)&=& c_2(d,N)  \sum_{i_2=0}^{d-2} \Big|  x_{i_2}^{(2)} +  x_{i_2+1}^{(2)}   \Big| ^2 \nonumber\\
&=& c_2(d,N) \bra{x^{(2)}} \frakF_2 \ket{x^{(2)}},
\eeqn
where $\textbf{x}_2=(x_0^{(2)},x_1^{(2)},\dots,x_{d-1}^{(2)})$,
$c_2(d,N)=\frac{1}{4}[(d-1)/d]^{N-2}$,  and $\frakF_2$ is a tridiagonal matrix given by (\ref{tridiagonal-x}) with $\alpha=1$, $\beta=1$, $g=1$. By the results of \cite{losonczi} we then conclude that the  coefficients of  $\ket{x^{(2)}}_{A_2}$, which are optimal at this step of the algorithm  are given by (again, disregarding possible phases)  
\beqn \label{optimal-coordinates}
\tilde{x}_{i_2}^{(2)}= \sqrt{\frac{2}{d}} \sin \frac{(2 i_2+1)\pi}{2d},\quad i_2=0,1,\dots,d-1.
\eeqn

In the next step, we substitute the found values for parties $A_1$ and $A_2$ leaving the parties $A_3,\dots, A_N$ untouched and obtain the following quantity to be maximized: 
\beqn
F_3(\textbf{x}_3)&=&c_3(d,N) \sum_{i_2,i_3=0}^{d-2}  \Big|  \tilde{x}_{i_2}^{(2)}x_{i_3}^{(3)} +  \tilde{x}_{i_2+1}^{(2)} x_{i_3+1}^{(3)}    \Big| ^2 \non
&=&  c_3(d,N)  \sum_{i_3=0}^{d-2}  \left\{ w_1|x_{i_3}^{(3)}|^2 + w_2 |x_{i_3+1}^{(3)}|^2 \right.\non 
&&\hspace{+1.5cm}+ \left.  w_3\left[ (x_{i_3}^{(3)})^*  x_{i_3+1}^{(3)} + x_{i_3}^{(3)}  (x_{i_3+1}^{(3)})^* \right]\right\}\non
&=& c_3(d,N)  \bra{x^{(3)}} \frakF_3 \ket{x^{(3)}},
\eeqn
where $c_3(d,N)=\frac{1}{4}[(d-1)/d]^{N-3}$, $\textbf{x}_3=(x_0^{(3)},x_1^{(3)},\dots,x_{d-1}^{(3)})$, and
\beqn
w_1 &=&  \sum_{i_2=0}^{d-2} \left(\tilde{x}_{i_2}^{(2)}\right)^2 =1-\frac{2}{d} \sin^2 \frac{\pi}{2d}=\frac{1}{d}\left(d-1 +\cos\frac{\pi}{d}    \right), \non  
w_2 &=&  \sum_{i_2=0}^{d-2} \left(\tilde{x}_{i_2+1}^{(2)}\right)^2 =1-\frac{2}{d} \sin^2 \frac{\pi}{2d}=\frac{1}{d}\left(d-1 +\cos\frac{\pi}{d}    \right), \non 
w_3 &=&  \sum_{i_2=0}^{d-2} \tilde{x}_{i_2}^{(2)} \tilde{x}_{i_2+1}^{(2)}  =\frac{1}{d} \left((d-1) \cos \left(\frac{\pi }{d}\right)+1\right),
\eeqn
and $\frakF_3$ is again a matrix of the form (\ref{tridiagonal-x}), in this case with $\alpha=w_1$, $\beta=w_2$, $g=w_3$.

This time, however,  we are not able to compute the eigenvalues of the corresponding matrix and in turn find the exact value of the GM.

The idea thus is to find an easily computable bound instead. With this aim take all $x$'s having coordinates as in (\ref{optimal-coordinates}). This clearly results in an upper bound on the GM and it is easy to see that it is of the  following form:
\beqn \label{bound-GM}
\frac{1}{4} \left( w_1^{N-1} +w_2^{N-1}+ 2w_3^{N-1} \right)=&&  \\
&&\hspace{-4cm}\frac{1}{2d^{N-1}} \left[ \left(d-1 +\cos\frac{\pi}{d}    \right)^{N-1} + \right. \non &&\hspace{-2cm}+ \left. \left((d-1) \cos \left(\frac{\pi }{d}\right)+1\right)^{N-1}   \right]. \nonumber
\eeqn
\begin{center}
		\begin{table}
			\begin{tabular}{rcccc} 	\hline\hline
				&&	$N\rightarrow$&&\\
				&                     3&        4   &         5&         6\\ \hline
				3   &                     $3/7$\; ($31/72$) & 0.490 (0.563) &     0.498 (0.660)  & 0.499 (0.733)     \\
				$d \downarrow$\quad	4 &  0.265 (0.266)  & 0.360 (0.364) &     0.432 (0.446)    & 0.475 (0.514)  \\
				5 &                  0.178 (0.179) &   0.250 (0.251)  &      0.311 (0.315)      & \; --- \;\;  (0.370) \\
					\hline\hline
			\end{tabular}	
			\caption{$E_{GM}(\calS_{2\times d^{N-1}}^{\pi/2})$ for several values of $d$ and $N$ obtained through a numerical optimization (for $N=3$ and $d=3,4$ the results are analytic). In the parentheses, we give the analytical upper bounds (\ref{bound-GM}).}
			\label{ges-general-gm}
		\end{table}
	\end{center}

In Table \ref{ges-general-gm}, we compare the obtained bound with the results of a numerical optimization of (\ref{key-problem}) in the case of $a=b$ for some values of $d$ and $N$.
 We see that there is a clear trend in the values: for a given $N$ the GM drops with the dimension $d$; on the other hand, for a given $d$ it grows with $N$ (probably tending to $0.5$). While the bound (\ref{bound-GM}) is very tight for $N=3$, we observe that it gets, as expected, much weaker when the number of parties increases. However, for a given $N$ it gets tighter with the increasing dimension $d$. We conclude by noting that for $d=3,4$ and $N=3$ it is possible to obtain an analytical value of the GM, we omit the detail here, though, as this is just a simple algebra.

\section{Lower bounds on the entanglement of a subspace in terms of SDP}\label{sdp-boundy}

Analytical computation of the subspace entanglement will usually be a very difficult problem (cf. \cite{Huang-NP}). Our result and the one of Ref. \cite{Branciard} seem to be notable exceptions. In particular, this will not be accessible for large systems or subspaces with complicated basis vectors. It is thus desirable to have at one's disposal easily computable bounds. We consider this problem in the present section.

The form of the minimization problem (\ref{przeksztalcenie}) directly allows us to bound the entanglement of a subspace from below using a relaxation involving instances of SDPs. Namely, the GM can be bounded as follows
\begin{eqnarray}\label{with-fully-product}
E_{GM}(\calS)&=&
\min_{\ket{\psi_{\mathrm{prod}}}}\langle\psi_{\mathrm{prod}}|\calP_{\calS}^{\perp}|\psi_{\mathrm{prod}}\rangle\non&=&\min_{\ket{\psi_{\mathrm{prod}}}}\tr[\calP_{\calS}^{\perp}\proj{\psi_{\mathrm{prod}}}]\nonumber\\
&\geq & \min_{\substack{\rho\geq 0\\\forall_K \rho^{T_K}\geq 0}}\tr[\calP_{\calS}^{\perp}\rho],
\end{eqnarray}
while for the GGM one has
\begin{eqnarray}\label{with-biproduct}
E_{GGM}(\calS)&=&\min_{\ket{\psi_{\mathrm{biprod}}}}\langle\psi_{\mathrm{biprod}}|\calP_{\calS}^{\perp}|\psi_{\mathrm{biprod}}\rangle\non&=&\min_{\ket{\psi_{\mathrm{biprod}}}}
\tr[\calP_{\calS}^{\perp}\proj{\psi_{\mathrm{biprod}}}]\nonumber\\
&\ge&\min_{K|\bar{K}}\left\{ \min_{\substack{\rho\geq 0\\\rho^{T_{K}}\geq 0}}\tr[\calP_{\calS}^{\perp}\rho]\right\}.
\end{eqnarray}
In the above, $\rho^{T_K}$ denotes the partial transpose of $\rho$ with respect to the bipartition $K|\bar{K}$. The idea behind this relaxation is that product states necessarily have a positive partial transpose (PPT) across the cuts with respect to which they are product, but these are not the only states having the PPT property, as there exist PPT entangled states,
and thus we enlarge the class of states over which the optimization should be performed.  This clearly results in lower bounds.

We should mention that the idea of using the set of PPT states in such optimization problems is not new and it has been very recently applied in \cite{GM-PPT-new} to find the GM of the Werner and the isotropic states (see also \cite{ent-depth}). 
We discuss the use of such approach for  computing bounds on the (G)GM of a class of multipartite states in further parts. As another remark, we also note that in this kind of problem, where minimization over (bi)product vectors is required, it is obviously possible to introduce intermediate relaxations in Eqs. (\ref{with-fully-product}) and (\ref{with-biproduct}) requiring $\rho$'s, over which the minima are taken, to be $k$--symmetric PPT extendible \cite{DPS,power-of-PPT} (see also \cite{EntCriteria-hierarchy} ).
Finally, let us note that in the relaxation one could use optimization over states, which remain  positive under a different positive map, not necessarily the transposition; this could  for example be the Breuer--Hall map \cite{Breuer,hall}. To make a better approximation of the set of (bi)separable states, this could also be performed over states which stay positive under a set of positive maps \cite{Lancien_2015}. All these relaxations are SDPs.

Note that if the bound for the GGM was zero this would signify that there
exists a PPT GME state supported on the GES under scrutiny. This, however, will not be the case for the examples we consider.

The bounds, from now on referred to as the {\it SDP bounds} and denoted $E_{(G)GM}^{SDP}$, have been computed in a few relevant cases and compared with the results obtained via direct minimization or the ones found in the literature to check their performance. We discuss the results below.

\subsection{Case 1: subspace $\calS_{2\times d^{N-1}}^{\pi/2}$ (Definition \ref{ges-def})}

Table \ref{Th1-porownanie} presents a comparison of the SDP bounds and the results obtained numerically (some of them analytically) for $\calS_{2\times d^{N-1}}^{\pi/2}$ introduced in Definition \ref{ges-def}.
There is a curious case of $d=3$, where we observe discrepancy between the analytic result and the SDP bound. In Appendix \ref{exemplary-state}, we give an exemplary class of states beating the analytical value for this case. This class may serve as a starting point of a future analysis toward the resolution of the problem of why a gap is observed here and, more generally, in some other setups as well.  In the remaining cases the difference between the results is $10^{-8}$ --  $10^{-12}$ depending on the scenario.

	\begin{center}
		\begin{table}[h!]
			\begin{tabular}{ccccc} \hline \hline
				{$d$} & $E_{GM}(\calS_{2\times d^2}^{\pi/2})$  & $E_{GM}^{SDP} (\calS_{2\times d^2}^{\pi/2})$ &  $E_{GGM}(\calS_{2\times d^2}^{\pi/2})$ &$E_{GGM}^{SDP}(\calS_{2\times d^2}^{\pi/2})$\\ \hline
				$3$   & 0.42857 & 0.41416 & 0.25000 & 0.25000   \\
				$4$  & 0.26543 & 0.26543& 0.14645 & 0.14645  \\
				$5$  & 0.17837  & 0.17837 & 0.09549 & 0.09549 \\
				$6$  & 0.12742  & 0.12742  & 0.06699 & 0.06699\\
				$7$ &  0.09530 & 0.09530  & 0.04952 & 0.04952 \\ 
				$8$  & 0.07384 & 0.07384 & 0.03806& 0.03806\\ \hline \hline
			\end{tabular}
			\caption{Comparison of the SDP bounds and the values of the (G)GM for $\calS_{2\times d^2}^{\pi/2}$ from Definition \ref{ges-def}. The values of $E_{GM}$ for $d=3,4$ have been obtained analytically, the remaining ones numerically. The values of the GGM are given by the formula $E_{GGM}(\calS_{2\times d^2}^{\pi/2})=\sin^2(\pi/(2d))$ (Theorem \ref{ges-ggm}). Except for a single case $d=3$ in the GM case the results match with a precision ranging from $10^{-8}-10^{-12}$.}  \label{Th1-porownanie}
		\end{table}
	\end{center}

\subsection{Case 2: subspace $\calQ_1^{N,d}$ (Theorem 2 of Ref. \cite{upb-to-ges})} \label{q1}

As the second example, we consider the subspace introduced in Theorem 2 of Ref. \cite{upb-to-ges}. Originally, it has only been  defined as a subspace of $(\cee{d})^{\otimes N}$ orthogonal to a continuous set of product vectors (a non--orthogonal unextendible product basis) with no explicit basis  given. In Appendix \ref{baza-th2}, we fill this gap and find that this GES --- in the present paper called $\calQ_1^{N,d}$ --- is spanned by the following (non--orthogonal) unnormalized vectors:
\beqn
&&\ket{i}_{A_1}\ket{p_{m-i}+k}_{\aaa{2}{k}}-\ket{i+1}_{A_1}\ket{p_{m-i-1}+k}_{\aaa{2}{k}}, \non
&& \quad m=1,2,\dots, d-2, i=0,1,\dots,m-1, \non
&&\ket{i}_{A_1}\ket{p_{d-i-1}}_{\aaa{2}{k}}-\ket{i+1}_{A_1}\ket{p_{d-i-2}}_{\aaa{2}{k}}, \non && \quad i=0,1,\dots,d-2,\\
&&\ket{i}_{A_1}\ket{p_{m-i}+k+1}_{\aaa{2}{k}}-\ket{i+1}_{A_1}\ket{p_{m-i-1}+k+1}_{\aaa{2}{k}}, \non
&& \quad m=d-1,d,\dots, 2(d-2), i=m-(d-2),\dots, d-2,\nonumber
\eeqn
where $p_i=i p_1$, $i=0,1,\dots, d-1$, with
$ p_1= \sum_{m=2}^Nd^{N-m}$.
The dimension of the GES is $\dim \calQ_1^{N,d}=d^N-(2d^{N-1}-1)$.

We have tested the SDP bounds on $\calQ_1^{3,d}$ for $d=3,4,5$  and the results are given in  Table \ref{sdp-th2} below. We observe that in some cases the SDP bounds most probably provide exact values; however, in the case of the GGM for $d=3$ the gap between the results is quite large.

\begin{table}[h!]
	\begin{tabular}{cccll} \hline \hline
		$d$ & $N$ & $\dim \calQ_1^{3,d}$ & $E_{GM}^{SDP}(\calQ_1^{3,d})$ & $E_{GGM}^{SDP}(\calQ_1^{3,d})$ \\ \hline
		$3$ & 3 & $10$ &$0.19022$ $(0.19036) $ &   $0.025078$  $(0.030844)$ \\
		$4$ & 3 & $33$ &$0.03696$   &   $0.000976$ $(0.001144)$   \\
	  $5$ & 3 & 76 & 0.00629 &  $0.000016$ ($0.000024$)   \\  \hline \hline
	\end{tabular}
	\caption{The SDP bounds on GM and GGM  of  $\calQ_1^{3,d}$. The parentheses give the results of a numerical minimization of $E_{GM}$ and $E_{GGM}$ correspondingly. No parentheses means that the values we have found are the same.}
	\label{sdp-th2}
\end{table}

\subsection{Case 3: subspace $\calQ_2^{N,d}$ (Theorem 3 of Ref. \cite{upb-to-ges})} \label{q2}

Our third test case is the GES --- here called $\calQ_2^{N,d}$ --- from Theorem 3 of Ref. \cite{upb-to-ges}. 
A basis for this GES has been given in \cite{upb-to-ges} in the case of qubits. We find a basis in the general case in Appendix \ref{baza-th3} to be:
\beqn
&&\ket{0}_{A_1}\left( \sum_{f=2}^{N} \ket{kd^{N-f}+m}_{\aaa{2}{k}}   \right)  - \ket{k}_{A_1}\ket{m}_{\aaa{2}{k}}, \\
&& \hspace{+0.5cm} k=1,\dots, d-1, \quad m=0,1,\dots, d^{N-1}-d^{N-2}-1.\nonumber
\eeqn
The dimension is $\dim \calQ_2^{N,d}=d^{N-2}(d-1)^2$.

\subsubsection{Qubits: $\calQ_2^{N,2}$ }

In the qubit case ($d=2$) the subspaces $\calS_{2\times 2^{N-1}}$ and $\calQ_1^{N,2}$ are equivalent and consist of a single state, which is locally unitarily equivalent to the GHZ state.  In contrast, $\calQ_2^{N,2}$ is nontrivial and its dimension is $2^{N-2}$.

We present the results in Table \ref{miary-kubity}.
In the case of $E_{GM}$ an agreement up to $12$ significant digits between the SDP bounds and the direct minimization has been observed.
For $E_{GGM}$, an  agreement to $6-8$ significant digits, depending on the case, is observed. We thus feel inclined to believe that the SDP bounds are exact values of the entanglement of the GESs in this case. Unfortunately, we have not been able to find an analytical proof of this fact, but it seems plausible that such a proof should rely on particular properties of systems with qubit subsystems.

\begin{table}[h!]
	\begin{tabular}{cccc} \hline \hline
		$N$  &$\dim \calQ_2^{N,2}$&  $E_{GM}(\calQ_2^{N,2})$ & {$E_{GGM}(\calQ_2^{N,2})$ } \\ \hline
		$3$  & 2 &$0.2640$  &  $0.07810$   \\
		$4$ & 4 & $0.1794$ &  $0.01637$     \\ 
		$5$  & 8 & $0.1213$   &   $0.00436$   \\
		$6$  & 16 &$0.0821$  &   $0.00099$   \\ \hline \hline
	\end{tabular}
	\caption{ The GM and GGM of the qubit subspaces $\calQ_2^{N,2}$. The SDP bounds and numerical results match with very high precision. }
	\label{miary-kubity}
\end{table}

 We observe that the entanglement in both cases drops with the dimension of a subspace. This clearly happens more rapidly for $E_{GGM}$ than for $E_{GM}$.

\subsubsection{Tripartite qudit case: $\calQ_2^{3,d}$}

The results for higher dimensions in the tripartite case ($N=3$) are collected in Table \ref{sdp-th3}.

We have found agreement for $E_{GM}$ to  $8-12$ significant digits indicating that the SDP bounds represent the true values in these cases. On the other hand, the results for $E_{GGM}$ suggest that this is not  true in this case.

\begin{table}[h!]
	\begin{tabular}{ccccl} \hline \hline
		{$d$} &$N$ & $\dim \calQ_2^{3,d}$ & $E_{GM}^{SDP}(\calQ_2^{3,d})$  &$E_{GGM}^{SDP}(\calQ_2^{3,d})$\\ \hline
		$3$  & $3$  & $12$ &$0.05856$   &   $ 4.8023 \cdot 10^{-3}$  $(4.8184 \cdot 10^{-3})$  \\
		$4$ & $3$ & $36$  &  $0.00753$ &  $1.2579 \cdot 10^{-4}$ $(1.2649 \cdot 10^{-4})$\\
			$5$ & 3 & 80 &  0.00124  & $2.2147 \cdot 10^{-6}$ $(2.2727 \cdot 10^{-6})$    \\ 
			 \hline \hline
	\end{tabular}
	\caption{The SDP bounds on $\calQ_2^{3,d}$ in the cases of  qutrits ($d=3$), ququarts ($d=4$), and ququints ($d=5$). In the parentheses we give the values obtained with a direct minimization over biproduct vectors. }
	\label{sdp-th3}
\end{table}

\subsection{Other subspaces from the literature}

We have also checked how the SDP bounds behave for some other subspaces
not considered either here or in \cite{upb-to-ges}.

First, we have verified that they reproduce the correct value of the entanglement of  any superpositions of either two of the following states: $\ket{W}$,  $\ket{\tilde{W}}=\sigma_x ^{\otimes 3} \ket{W}$, and $\ket{GHZ}$  (cf. \cite{WeiGoldbart,Blasone2008}). In particular, this means that they recover the results concerning the measure for each of the states above and give the correct values for two--dimensional GESs spanned by either two states from above.

Second, we have considered the antisymmetric subspace. The antisymmetric subspace, $\calA_{d,N}$, of an $N$--partite Hilbert space $(\mathbb{C}^d)^{\otimes N}$, $d\ge N$, is the ${d \choose N}$ dimensional subspace spanned by the states acquiring the minus sign when any odd permutation of the parties is done.
It is easy to realize that  $\calA_{d,N}$ is indeed a GES.
By the result of \cite{geometric-hayashi} the geometric measure of the antisymmetric subspace equals
\beqn\label{gm-anti}
E_{GM}(\calA_{d,N})=1-\frac{1}{N!},
\eeqn
that is, it is independent of the local dimension $d$.
Numerically, we found that in the case of the GGM this property also holds, and the value is:
\beqn\label{ggm-anti}
E_{GGM}(\calA_{d,N})=1-\frac{1}{N}.
\eeqn
Unfortunately, we have only managed to consider the SDP bounds for small systems but found they agree in these cases with (\ref{gm-anti}) and (\ref{ggm-anti}).

\section{Entanglement properties of states constructed from a GES} \label{different-methods}

We now approach the problem of quantifying entanglement of subspaces from a different perspective. Namely, we consider the entanglement of states:
\beqn \label{noisy-ges}
\varrho_{\calG}(p)=(1-p) \frac{\calP_{\calG}}{d_{\calG}}+p \frac{\jedynka_{D}}{D}, \quad D=\Pi_i d_i,
\eeqn
 where $\calG$ is a GES, $\calP_{\calG}$ is the projection onto $\calG$, and $d_{\calG}=\mathrm{rank}\;\calP_{\calG}$; we will call them {\it noisy GES states}.
 
We will consider two extreme cases to characterize the entanglement of a subspace $\calG$:

(i) the (G)GM of a noiseless state, $E_{(G)GM}(\varrho_{\calG}(0))$ (cf. \cite{Bracken,geometric-hayashi});

(ii) the white--noise tolerances $p^*_{gme}$ and $p^*_{ent.}$, that is, the threshold probabilities below which the state (\ref{noisy-ges}) is certainly GME or entangled, respectively.   

The parameters above, although not strictly entanglement measures, seem natural  alternative quantifiers of entanglement of $\calG$ apart from  $E(\calG)$ itself. Moreover, $E(\varrho_{\calG}(0))$, by (\ref{state-vs-subspace}),  provides an upper bound on the entanglement of $\calG$, while  the critical probability $p^*$ is of particular importance from an experimental point of view.

We have considered the problems above using several different methods. In Secs. \ref{witnes}-\ref{algo} we briefly describe them and in Sec. \ref{results} we gather and analyze the obtained results. 
In our study we concentrate on the tripartite ($N=3$) case, which is dictated by the computational power available to us. Our test GESs are those considered extensively so far in the paper: $\calS$, $\calQ_1$, and $\calQ_2$.

\subsection{Connection with entanglement witnesses}\label{witnes}

Here we establish a connection between the entanglement of a subspace and entanglement witnesses. This will allow, in particular, for a simple analytical estimation of the white--noise tolerance of $\rho_{\calG}(0)=\calP_{\calG}/d_{\calG}$.

In \cite{upb-to-ges} we noted that the genuine entanglement of $\rho_{\calG}(0)$ from (\ref{noisy-ges})
can be witnessed by the following entanglement witness \cite{LewensteinWitnesses2001}:
\beqn \label{general-witness}
W_{\mathrm{GES}}=\frac{1}{(1-\epsilon_{gme}) D-d_{\calG}} \Big[(1-\epsilon_{gme}) \jedynka_{D}-\calP_{\calG}\Big]
\eeqn
with 
\beq \label{epsilon}
\epsilon_{gme}\equiv\min _{\ket{\psi_{\mathrm{biprod}}}\;}\bra{\psi_{\mathrm{biprod}}}P_{\calG}^{\perp}\ket{\psi_{\mathrm{biprod}}}.
\eeq

Now, comparing (\ref{epsilon}) with (\ref{przeksztalcenie}), we immediately infer that:
\beq
\epsilon_{gme}= E_{GGM}(\calG),
\eeq
providing a direct link between these two notions. This can be viewed as a generalization of the observation made already  in Ref. \cite{WeiGoldbart} for  pure states.

In fact, the witness $W_{GES}$ detects {\it all} states supported on $\calG$, not only of $\rho_{\calG}(0)$. However, it comes with a price: it gives a constant value on all such states, and as such cannot be used to compare their entanglement. This behavior can be attributed to the non-optimality of the witnesses.

Using the witnesses we can bound the white--noise tolerance $p^*$ of {\it any} states with the support in a GES.
	With this aim consider more general states than in (\ref{noisy-ges}):
\beqn \label{gamma-p}
\gamma_{\calG}(p)=(1-p) \sigma_{\calG}+p \frac{\jedynka_{D}}{D}, \quad \mathrm{supp}(\sigma_{\calG}) \subseteq \calG.
\eeqn
Evaluating their mean value with $W_{\mathrm{GES}}$ we get:
\beqn 
\tr  W_{\mathrm{GES}}\gamma_{\calG}(p) =\frac{1}{(1-\epsilon_{gme}) D-d_{\calG}} \left(\frac{p (D-d_{\calG})}{D}-\epsilon_{gme}\right), \non
\eeqn
from which we obtain that $\gamma_{\calG}(p)$, with $\varrho_{\calG}(p)$ as a special case, is certainly GME at least in the region:
\beq \displaystyle
\label{zakres-GME}
p < p^{*witn.}_{gme}\equiv\frac{D\epsilon_{gme}}{D-d_{\calG}}.
\eeq
For example, in the case  of $\calS_{2\times d^2}^{\pi/2}$ (Section \ref{th1-generalized}) we obtain an analytical  noise tolerance threshold:
\beq
p_{gme}^{*witn.}(\calS_{2\times d^2}^{\pi/2})= \frac{2d^2}{d^2+2d-1} \sin^2\frac{\pi}{2d}.
\eeq
It can be seen that the value of the white--noise tolerance predicted by this approach drops with an increasing local dimension $d$; as we will see later this happens faster than for the actual value of the white--noise tolerance.

The same reasoning can be applied to find an estimate on the value of $p^{*witn.}_{ent.}$, below which a state is certainly entangled.
To achieve this, instead of taking biproduct vectors in (\ref{epsilon}) we  take fully product ones to define $\epsilon_{ent}=E_{GM}(\calG)$ and use this quantity to construct the  witness just as in (\ref{general-witness}). The white--noise tolerance is simply given by:
\beq \displaystyle
\label{zakres-ent}
p^{*witn.}_{ent.}\equiv\frac{D\epsilon_{ent.}}{D-d_{\calG}}.
\eeq

\subsection{PPT mixtures}\label{mixtures}

The next method of detecting and quantifying genuine entanglement of states is the PPT mixtures technique from \cite{taming}. The idea is to check with an SDP whether a given state is a convex combination of PPT states. If it is not, the state is necessarily GME.
We  briefly recall for completeness the method below. 

Given a multipartite state $\rho$ one solves the following optimization problem:
\beqn
&&\min \tr (W\rho) \\
&&\mathrm{s.t.}\; 
W=P_K+Q_K^{T_K}, \; 0 \le P_K \le \jedynka, \;  0 \le Q_K \le \jedynka, \forall_{K|\bar{K}}, \non\label{warunki-swiadek}
\eeqn
where the condition for $W$ holds for all bipartitions $K|\bar{K}$.
The found witness $W$ is called fully decomposable and the function 
\beqn
E_{ppt}(\rho)=-\tr (W\rho)
\eeqn%
 is an entanglement monotone. Its non--zero value signifies genuine entanglement of the state.

A version of this problem, instead of the conditions (\ref{warunki-swiadek}), assumes that for all $K$ it only holds that $W^{T_K}\ge 0$, i.e., $P_K=0$. We then talk about fully PPT witnesses and they clearly are weaker. The monotone in this case will be denoted $E_{ppt}^{fully}$.

\subsection{Bounds on the (G)GM of states using PPT relaxations}\label{ppt-fidelity-relax}

SDP is useful in yet another way. Namely, exploiting the relation between the (G)GM of a state and the fidelity one can put forward simple SDP bounds on the measure and in turn have numerically computable estimates on the entanglement of the states under scrutiny.

The geometric measures have been shown to be directly related to the fidelity $F(\rho,\sigma)=\tr\sqrt{\sqrt{\rho}\sigma\sqrt{\rho}} $ through the following formulas \cite{Streltsov-gm-fidelity,Regula2018}:
\beqn
&&E_{GM}(\rho)=1-\max_{\sigma\; \mathrm{fully}\; \mathrm{sep.} }F^2(\rho,\sigma),\\
&&E_{GGM}(\rho)=1-\max_{\sigma\;  \mathrm{biseparable} }F^2(\rho,\sigma), 
\eeqn
where the maximization is over fully separable and biseparable states for the GM and GGM, respectively.

Such representations allow us  again to use certain relaxations. Precisely, we approximate the set of  separable states with the set of PPT states (across relevant cuts) and the set of biseparable states with the set of states which are mixtures of PPT states (as already discussed different forms of relaxations could be used). We then obtain the  bounds (cf. \cite{GM-PPT-new}):
 \beqn
 &&\hspace{-0.9cm}E_{GM}(\rho) \ge  1-\max_{\substack{\sigma\geq 0\\  \forall_K \sigma^{T_K}\geq 0}} F^2(\rho,\sigma)=:E_{GM}^{F}(\rho), \\
 &&\hspace{-0.9cm}E_{GGM}(\rho) \ge 1- \max_{\sigma\;\mathrm{PPT}\; {mixture}}  F^2(\rho,\sigma)  =:E_{GGM}^{F}(\rho).\label{GGM-fidelity-relaxation}
 \eeqn
We will later refer to these bounds as the {\it fidelity relaxation bounds}.
The value of the relaxations lies again in the fact that the fidelity can be efficiently computed with an instance of an SDP. Precisely, the fidelity $F(\rho,\sigma)$ is computed as follows \cite{Watrous}:
\beqn \label{sdp-fidelity}
&&\max\; \frac{1}{2} \tr X +\frac{1}{2} \tr  X^{\dagger}, \\
&&\mathrm{s.\; t.} 
\left( \begin{array}{cc}
  \rho	&  X\\
X^{\dagger}	& \sigma\\
\end{array}     \right) \ge 0 .
\eeqn
\subsection{Algorithmic approximation of the (G)GM}\label{algo}

Finally, we have applied the algorithm found in \cite{algorithm-gm} to approximate numerically the (G)GM of the relevant states and find their white--noise tolerances (we note that originally the algorithm is designed for the GM but a simple modification allows us to use it to approximate the GGM). In principle, since the algorithm only requires solving eigenproblems and finding the singular value decomposition of certain matrices it is easy to implement. 
However, it requires decompositions of density matrices into ensembles with $(d_1d_2 \cdots \dots \cdot d_N)^2$ terms, which quickly becomes intractable by a desktop computer. To get around this problem one needs to use smaller, i.e., not optimal in this respect, ensembles.  Our experience gained for smaller problems shows that if the number of terms in an ensemble is not unreasonably small and the precision parameter set in the algorithm is very small the results appear to be accurate. Nevertheless, one needs to keep in mind the limitations of the approach when comparing the numbers and treat them with care. Interestingly, this was not an issue for the computation of the  (G)GM of the normalized identity on $\calS^{\pi/2}_{2\times d^2}$, as the algorithm quickly converges even for small ensembles. This could be attributed to the nice structure of the basis vectors for this subspace which translates into less demanding computations.

\subsection{ Results obtained with methods from  sections \ref{witnes}-\ref{algo}}\label{results}

The results obtained with the aid of the methods described above are collected in Tables \ref{ent-projection}-\ref{noise-tolerance}.

Before we move to a detailed discussion of the results, we must emphasize, that  our aim here is not to compare the subspaces of different types; this would make little sense since they are of different dimensionalities for the same parameters. Our goal is rather to compare different methods for a given subspace type and see how the properties change with an increasing local dimension within a subspace class. 

\subsubsection{Entanglement of $\calP_{\calG}/d_{\calG}$}

Let us start with the entanglement of the noiseless GES states $\varrho_{\calG}(0)=\calP_{\calG}/d_{\calG}$. The results are presented in  Table \ref{ent-projection}.  For reference, the entanglement of the corresponding subspace, $E_{(G)GM}(\calG)$, has also been given. We have performed our calculations on a desktop computer and some problems turned out to be too big; these are marked with a long dash in the table.

\begin{widetext}
	\begin{center}
		\begin{table}
			\begin{tabular}{ccccc} \hline \hline				Entanglement $E$	 & $d$   & $\calS^{\pi/2}_{2\times d^2}$ (Sec. \ref{th1-generalized}) &$\calQ_1^{3,d}$ (Sec. \ref{q1} ) & $\calQ_2^{3,d}$ (Sec. \ref{q2})\\ \hline
				&2 & 0.5000 & 0.5000 & \:\:\:\:  0.0781 / 0.2640   \\
				& 3  & 0.2500 / 0.42857 &  0.030844 / 0.19036     &    0.0048184 / 0.05856\\
				$E_{GGM}(\calG)$/$E_{GM}(\calG)$   & 4  &  0.1465 / 0.26543  & 0.001144 / 0.03696      & 0.0001265 /  0.00753 \\
				(Sections \ref{th1-generalized}, \ref{q1},\ref{q2},      & 5 & 0.0955  / 0.17837  & 0.000024 / 0.00629   & 0.0000023 / 0.00124 \\
				cf. Tables \ref{Th1-porownanie}, \ref{sdp-th2}, \ref{sdp-th3})	& 6  &   0.0670 / 0.1274 & --- & --- \\
				& 7  &  0.0495 / 0.0953 & --- & --- \\
				& 8  & 0.0381  / 0.0738 & --- & --- \\
				\hline
				&2& 0.5000  &  0.5000  &    0.1832 / 0.0942 \\
				& 3  & 0.3008 / 0.2253  & 0.0951 / 0.0774 & 0.0600 / 0.0253 \\
				$E_{ppt}$/$E_{ppt}^{fully}$   & 4  & 0.1905 / 0.1361   &  0.0375 /  0.0219 & 0.0246 / 0.0087 \\
				(Section \ref{mixtures})      & 5 & 0.1347 / 0.0902  & --- & --- \\
				& 6  & 0.1012 / 0.0641 & --- & --- \\
				& 7  &\;\; --- \;\; / 0.0479  & --- & --- \\
				& 8  & \;\; ---\;\; /  0.0372  & --- & --- \\
				\hline
			   &2& 0.5000  &  0.5000  &  0.0875 /  0.2801   \\
				& 3  & 0.2286 / 0.4150  & 0.0663 / 0.2297 & 0.0293 / 0.1380 \\
				$E_{GGM}^{F}$/$E_{GM}^{F}$   & 4  & 0.1316 / 0.3056  &  0.0231 / 0.1410 & 0.0128 / 0.0859  \\
				(Section \ref{ppt-fidelity-relax})      & 5 &   0.0872 / 0.2344  &--- & --- \\
				& 6  & 0.0625  / 0.1900  & --- & --- \\
				& 7  & \;\; ---\;\;\; / 0.1597  & --- & --- \\
				\hline
					&2& 0.5000  &  0.5000  &  0.0962 /  0.2801  \\
				& 3  & 0.2500 / 0.4375 & 0.0816 / 0.2297  &   0.0378 / 0.1524 \\
				$E_{GGM}^{algor.}$/$E_{GM}^{algor.}$   & 4  & 0.1667  / 0.3056  &  0.0345 / 0.1491& 0.0183 / 0.0980 \\
				(Section \ref{algo})      & 5 &  0.1250 / 0.2344 &  0.0172 / 0.1065& 0.0111 / 0.0712 \\
				& 6  &   0.1000 / 0.1900  & --- & --- \\
				& 7  &  0.0833 / 0.1597 & --- & --- \\
				\hline \hline
			\end{tabular}
			\caption{Entanglement of the normalized identity on a GES as given by different methods. For reference, the entanglement of subspaces is given. A long dash indicates problems too large for a desktop. See the main text for a discussion. }
			\label{ent-projection}
		\end{table}
	\end{center}
\end{widetext}

\noindent\textit{PPT mixtures (Section \ref{mixtures})}.
We see an interesting behavior for the subspace $\calS_{2\times d^{2}}^{\pi/2}$, namely, $E_{ppt}^{fully} < E_{GGM}(\calS)< E_{ppt}$,  with $E_{ppt}^{fully}$ approaching the value of $E_{GGM}(\calS)$ for an increasing local dimension. This is, however, not observed for subspaces $\calQ_i^{3,d}$, in which cases both $E_{ppt}$ and $E_{ppt}^{fully}$ exceed quite largely the entanglement of the subspace.

\noindent\textit{Fidelity relaxation bounds on the (G)GM (Section \ref{ppt-fidelity-relax})}.
 Let us start with $\calS_{2\times d^{N-1}}^{\pi/2}$, $d \ge 3$, and the case of the GGM. We observe that the bounds are trivial as the better ones are simply provided by  $E_{GGM}(\calS)$ [cf. Eq. (\ref{state-vs-subspace})]; nevertheless, $E_{GGM}^F$ approaches the latter from below for the increasing local dimension.
in the case of the GM,  we see that for $d=3$ the bound is useless for the same reason as above, i.e., it is below the value of $E_{GM}(\calS)$. The remaining values, however, are nontrivial bounds as they are above the entanglement of the subspace. 
In this case, we also see that the gap between $E_{GM}^{F}(\varrho_{\calS}(0))$ and $E_{GM}(\calS)$ gets larger with the increasing $d$, meaning that the former bound becomes more useful.
More importantly, however, it appears that in the case $d>3$ the values of the bound are in fact the values of $E_{GM}(\varrho_{\calS}(0))$. We comment on this issue again below.

\noindent For subspaces  $\calQ_i^{3,d}$ with  $d\ge 5$ the problem is too large for the desktop we have used and we have only managed to treat the cases $d=2,3,4$. These scarce data show that the bounds are non trivial as they exceed the entanglement of the subspaces.

\noindent\textit{Algorithmic approximation of the (G)GM (Section \ref{algo})}.
Let us again begin with $\calS_{2\times d^{2}}^{\pi/2}$. 
We believe that the values obtained with the algorithm are enough to infer that the PPT bounds on the GM  discussed above are the exact values (we have observed agreement to 8-10 significant figures) except for the already identified peculiar case $d=3$. in the case of the GGM on the other hand, we observe that it equals the entanglement of the subspace for $d=3$, but for larger $d$ the gap between these values grows.
The results also provide strong evidence that the value of the GGM equals
${1}/(2d-2)$. 
For the remaining two subspaces the data are again limited but  show that $E_{(G)GM}^F<E_{(G)GM}^{algor.}$ except for two nontrivial cases: the GGM in the case of  $\calQ_1^{3,3}$ and the GM in the case of $\calQ_2^{3,2}$. We have not been able to identify either necessary or sufficient conditions on subspaces for the equality to hold and this is left as an open problem.

\subsubsection{The white--noise tolerance of  $\calP_{\calG}/d_{\calG}$}

Let us now move to the case of the white--noise tolerance of the noiseless states from (\ref{noisy-ges}) $\varrho_{\calG}(0)=\calP_{\calG}/d_{\calG}$.

\noindent \textit{Witnesses (Section \ref{witnes})}.
Not surprisingly, the white--noise tolerance predicted by this method is not high, with particularly low values for subspaces $\calQ_1^{3,d}$ and $\calQ_2^{3,d}$ --- this is the price one pays for the generality of the bound, i.e., the fact that it works for any state of the form (\ref{gamma-p}).  Specifying the states $\rho$ to be the identities on $\calG$ we are able to increase these values significantly using different methods.

\noindent\textit{PPT mixtures (Section \ref{mixtures}), fidelity relaxation bounds on the GGM (Section \ref{ppt-fidelity-relax}), and the algorithmic computation of the GGM (Section \ref{algo})}. Contrary to other methods, the PPT mixtures approach only deals with genuine entanglement of states and, clearly, the values obtained with its aid are the same as the ones obtained with the relaxation (\ref{GGM-fidelity-relaxation}).
The PPT mixtures have proved useful in  improving estimations of the white--noise tolerance for several classes of genuinely entangled states \cite{taming} and we thus expected to improve significantly using them on the values obtained with the witnesses. This was indeed the case, but the improvement is more siginificant for $\calQ_i$'s than for $\calS_{2\times d^2}$. In the case of $\calS_{2\times d^2}$, it may be argued that the threshold tolerances predicted by the approach tend to a limit value, which probably lies between $0.18-0.20$. It is probably also the case for $Q_i$'s but the data are too limited to estimate the limits. An analogous behavior is observed for the values found with the algorithm for the GGM. In this case, however, the limit value for $\calS_{2\times d^2}$ lies significantly above the one predicted by the PPT mixtures and is most likely above $0.44$. We note that apart from the trivial qubit cases  $\calS_{2\times 2^2}$ and $\calQ_1^{3,2}$, the methods give the same threshold value only for $\calQ_2^{3,2}$. We conclude by observing that genuine entanglement of subspaces $Q_1$ and $\calQ_2$ displays very low white--noise tolerance.

\noindent\textit{Fidelity relaxation bounds on the GM (Section \ref{ppt-fidelity-relax}) and the algorithmic computation of the GM (Section \ref{algo})}. 
It turns out that the threshold values obtained using the PPT relaxations on the fidelity are the same as from the algorithm  for all three subspaces considered in the paper. As previously, it appears there are some limit values for the thresholds, which are significantly higher than the ones for GME, this is particularly visible for $Q_i$'s.

 The meaning of the threshold values obtained with the algorithm  is that in the region $0\le p \le p^{*algor.}_{gme}$ a state is GME, whenever $p^{*algor.}_{gme} < p \le p^{*algor.}_{ent}$ a state is entangled but the entanglement is not genuine, i.e., a state is biseparable but not fully separable, and, finally, above $p^*_{ent}$ a state is fully separable.
Concluding, let us observe that for all subspaces it holds that $p^*_{gme} < p^*_{ent.}$ with gaps being quite large. This probably is a generic behavior and it would thus be interesting to find an example of a subspace for which both values of the white--noise tolerance are equal. 

\begin{widetext}
	\begin{center}
\begin{table}
	\begin{tabular}{ccccc} \hline \hline
	Noise tolerance $p^*$	 & $d$   & $\calS^{\pi/2}_{2\times d^2}$ (Sec. \ref{th1-generalized}) &$\calQ_1^{3,d}$ (Sec. \ref{q1} ) & $\calQ_2^{3,d}$ (Sec. \ref{q2})\\ \hline
		&2&  0.571 / 0.571 & \: 0.571 / 0.571   & \;\;\;0.1041 / 0.352   \\
		   &  3 & 0.321 / 0.551  & 0.0490 / 0.302& \;\;0.0087 / 0.105  \\
	$p_{gme}^{*witn.}$/	$p_{ent.}^{*witn.}$ &  4 & 0.204 / 0.369  & 0.0024 / 0.076 & \hspace{0.1cm} 0.0003 / 0.017 \\
(Section \ref{witnes})  &  5 & 0.140 / 0.262  & \hspace{-0.2cm}$6\cdot 10^{-5}$ / 0.016 &  \hspace{-0.2cm}$6.3\cdot 10^{-6}$ / 0.0034 \\
    & 6  & 0.103 / 0.195  & --- & --- \\
    \hline
  	&2& 0.571 & 0.571   &  0.265  \\
  & 3  & 0.409 & 0.224 & 0.128  \\
	$p_{gme}^{*ppt}$   & 4  &  0.300 & 0.126 & 0.076 \\
	(Section \ref{mixtures})      & 5 & 0.243  & --- & --- \\
	       & 6  & 0.212 &--- & --- \\
	       \hline
	       	&2& 0.571 / 0.799  & 0.571 / 0.799  &  0.265 / 0.630  \\
	      & 3  &  0.409 / 0.692  & 0.224 / 0.653  &  0.128 / 0.582 \\
	     $p_{gme}^{*F}$/$p_{ent.}^{*F}$   & 4  & 0.300  / 0.639  & 0.126 / 0.612 & 0.076 / 0.577   \\
	    
	    (Section \ref{ppt-fidelity-relax})      & 5 &  0.243 / 0.609  &--- & --- \\
	      & 6  & 0.212  /  0.590   & --- & --- \\
	      \hline
	      	&2& 0.571 / 0.799 & 0.571 / 0.799 &  0.265 / 0.630 \\
	      & 3  & 0.506 / 0.692   &  0.254 / 0.653   &   0.145 / 0.582 \\
	     $p_{gme}^{*algor.}$/$p_{ent.}^{*algor.}$   & 4  & 0.473  / 0.639  &  0.152 / 0.612 &  0.133 / 0.577 \\
	     (Section \ref{algo})      & 5 &  0.457  / 0.609 &--- & --- \\
	     & 6  &  0.449 /  0.590 & --- & --- \\
		\hline \hline
	\end{tabular}
	\caption{The white--noise tolerance estimates for different subspaces as given by various methods. A long dash indicates problems too large for a desktop. See the main text for a discussion. }
	\label{noise-tolerance}
\end{table}
\end{center}
\end{widetext}

\section{Conclusions and outlook}\label{konkluzje}

We have considered the problem of quantification of entanglement of genuinely entangled subspaces (GESs), that is subspaces  composed only of genuinely multiparty entangled (GME) states,  mainly using the (generalized) geometric measure [(G)GM] of entanglement. This has been done from three qualitatively different perspectives exploiting both analytical and numerical methods. The main one has used the definition of the subspace entanglement in terms of the least entangled pure state from the subspace. We have proposed an analytical method to compute it and provided an easily implementable semidefinite program (SDP) to lower--bound it. We have observed that in many cases these two methods agreed. In particular, they reproduced, except for the tripartite case with $d=3$, the same results for a new class of a GES, $\calS_{2\times d^{N-1}}^{\theta}$, introduced in the paper.
In the second approach, we have asked about the entanglement of a state being a normalized projection onto a GES. Here, we have exploited the method of the PPT mixtures, but also used certain SDP relaxations and a direct numerical algorithm for approximating the (G)GM. Comparison of the latter two  have  revealed agreement of the methods for  $\calS_{2\times d^{N-1}}^{\theta}$ in the case of the GM except the curious case of $N=3$ and $d=3$. Finally, in the third approach, we have considered how much of the white--noise such a normalized projection tolerates before the state gets fully-- or bi--separable. In addition to the methods mentioned above we have also used the established here connection of the problem with entanglement witnesses.  We have  observed that in the case of the GGM the latter method predicted the lowest values of the white--noise tolerance, while the PPT mixtures intermediate ones in relation to the ''exact'' values predicted by the algorithm. in the case of the GM, the values obtained from the witnesses are again the lowest, but the ones from PPT relaxation on the fidelity and the algorithm match.

The results of the present paper provoke several questions and suggest future research directions.

From a specialized point of view, it would be interesting to identify conditions under which the SDP bounds on the subspace entanglement reproduce the true values of the latter. In particular, one should look carefully into the case of the subspace $\calS_{2\times d^{N-1}}^{\theta}$, where only a single case of $N=3$ and $d=3$ gives differing results,  the qubit  subspace $\calQ_2^{N,2}$ (or, possibly, in general the qubit case), and the antisymmetric subspace.
It would also be interesting to consider other ways to quantify entanglement of GESs, such as the average entanglement or the maximal entanglement of a vector drawn from a subspace, as well as using other entanglement measures. A natural direction regarding the noise tolerance of GESs is to consider their entanglement robustness to local noise, which is relevant for scenarios with the distribution of particles in networks. The latter is the subject of ongoing research \cite{miniatura}.

From a more general perspective, as entangled subspaces play important roles in different areas of quantum information theory, it is desirable to construct more examples of GESs with analytically computed properties. In particular, it would be of interest to have examples of highly, but not maximally, entangled  subspaces and investigate the possibility of their application in (approximate) QEC.

\section{Acknowledgements }

M.D. is partially supported by National Science Centre (Poland) through the MINIATURA 2 program (2018/02/X/ST2/01448). M.D. acknowledges the hospitality of ICFO -- The Institute of Photonic Sciences, where part of this work has been done. R.A. acknowledges the support
from the Foundation for Polish Science through the First TEAM project (First TEAM/2017-4/31) cofinanced
by the European Union under the European Regional Development Fund. We thank M. Studzi\'nski for pointing out to us Ref. \cite{losonczi}. 
To solve SDPs we used CVX, a package for specifying and solving convex programs \cite{cvx,GrantBoyd08}, accompanied by QETLAB:  A {MATLAB} toolbox for quantum entanglement \cite{qetlab}.

\bibliography{cytacje4}

\begin{thebibliography}{61}%
\makeatletter
\providecommand \@ifxundefined [1]{%
 \@ifx{#1\undefined}
}%
\providecommand \@ifnum [1]{%
 \ifnum #1\expandafter \@firstoftwo
 \else \expandafter \@secondoftwo
 \fi
}%
\providecommand \@ifx [1]{%
 \ifx #1\expandafter \@firstoftwo
 \else \expandafter \@secondoftwo
 \fi
}%
\providecommand \natexlab [1]{#1}%
\providecommand \enquote  [1]{``#1''}%
\providecommand \bibnamefont  [1]{#1}%
\providecommand \bibfnamefont [1]{#1}%
\providecommand \citenamefont [1]{#1}%
\providecommand \href@noop [0]{\@secondoftwo}%
\providecommand \href [0]{\begingroup \@sanitize@url \@href}%
\providecommand \@href[1]{\@@startlink{#1}\@@href}%
\providecommand \@@href[1]{\endgroup#1\@@endlink}%
\providecommand \@sanitize@url [0]{\catcode `\\12\catcode `\$12\catcode
  `\&12\catcode `\#12\catcode `\^12\catcode `\_12\catcode `\%12\relax}%
\providecommand \@@startlink[1]{}%
\providecommand \@@endlink[0]{}%
\providecommand \url  [0]{\begingroup\@sanitize@url \@url }%
\providecommand \@url [1]{\endgroup\@href {#1}{\urlprefix }}%
\providecommand \urlprefix  [0]{URL }%
\providecommand \Eprint [0]{\href }%
\providecommand \doibase [0]{http://dx.doi.org/}%
\providecommand \selectlanguage [0]{\@gobble}%
\providecommand \bibinfo  [0]{\@secondoftwo}%
\providecommand \bibfield  [0]{\@secondoftwo}%
\providecommand \translation [1]{[#1]}%
\providecommand \BibitemOpen [0]{}%
\providecommand \bibitemStop [0]{}%
\providecommand \bibitemNoStop [0]{.\EOS\space}%
\providecommand \EOS [0]{\spacefactor3000\relax}%
\providecommand \BibitemShut  [1]{\csname bibitem#1\endcsname}%
\let\auto@bib@innerbib\@empty
\bibitem [{\citenamefont {T\'oth}(2012)}]{toth-metro}%
  \BibitemOpen
  \bibfield  {author} {\bibinfo {author} {\bibfnamefont {G.}~\bibnamefont
  {T\'oth}},\ }\href {\doibase 10.1103/PhysRevA.85.022322} {\bibfield
  {journal} {\bibinfo  {journal} {Phys. Rev. A}\ }\textbf {\bibinfo {volume}
  {85}},\ \bibinfo {pages} {022322} (\bibinfo {year} {2012})}\BibitemShut
  {NoStop}%
\bibitem [{\citenamefont {Rossi}\ \emph {et~al.}(2013)\citenamefont {Rossi},
  \citenamefont {Bru\ss{}},\ and\ \citenamefont {Macchiavello}}]{grover-gme}%
  \BibitemOpen
  \bibfield  {author} {\bibinfo {author} {\bibfnamefont {M.}~\bibnamefont
  {Rossi}}, \bibinfo {author} {\bibfnamefont {D.}~\bibnamefont {Bru\ss{}}}, \
  and\ \bibinfo {author} {\bibfnamefont {C.}~\bibnamefont {Macchiavello}},\
  }\href {\doibase 10.1103/PhysRevA.87.022331} {\bibfield  {journal} {\bibinfo
  {journal} {Phys. Rev. A}\ }\textbf {\bibinfo {volume} {87}},\ \bibinfo
  {pages} {022331} (\bibinfo {year} {2013})}\BibitemShut {NoStop}%
\bibitem [{\citenamefont {Epping}\ \emph {et~al.}(2017)\citenamefont {Epping},
  \citenamefont {Kampermann}, \citenamefont {macchiavello},\ and\ \citenamefont
  {Bru{\ss}}}]{Epping-qkd}%
  \BibitemOpen
  \bibfield  {author} {\bibinfo {author} {\bibfnamefont {M.}~\bibnamefont
  {Epping}}, \bibinfo {author} {\bibfnamefont {H.}~\bibnamefont {Kampermann}},
  \bibinfo {author} {\bibfnamefont {C.}~\bibnamefont {macchiavello}}, \ and\
  \bibinfo {author} {\bibfnamefont {D.}~\bibnamefont {Bru{\ss}}},\ }\href
  {\doibase 10.1088/1367-2630/aa8487} {\bibfield  {journal} {\bibinfo
  {journal} {New Journal of Physics}\ }\textbf {\bibinfo {volume} {19}},\
  \bibinfo {pages} {093012} (\bibinfo {year} {2017})}\BibitemShut {NoStop}%
\bibitem [{\citenamefont {Osterloh}(2014)}]{Osterloh-2014}%
  \BibitemOpen
  \bibfield  {author} {\bibinfo {author} {\bibfnamefont {A.}~\bibnamefont
  {Osterloh}},\ }\href {\doibase 10.1088/1751-8113/47/49/495301} {\bibfield
  {journal} {\bibinfo  {journal} {Journal of Physics A: Mathematical and
  Theoretical}\ }\textbf {\bibinfo {volume} {47}},\ \bibinfo {pages} {495301}
  (\bibinfo {year} {2014})}\BibitemShut {NoStop}%
\bibitem [{\citenamefont {Goyeneche}\ and\ \citenamefont {\ifmmode~\dot{Z}\else
  \.{Z}\fi{}yczkowski}(2014)}]{goyeneche-karel}%
  \BibitemOpen
  \bibfield  {author} {\bibinfo {author} {\bibfnamefont {D.}~\bibnamefont
  {Goyeneche}}\ and\ \bibinfo {author} {\bibfnamefont {K.}~\bibnamefont
  {\ifmmode~\dot{Z}\else \.{Z}\fi{}yczkowski}},\ }\href {\doibase
  10.1103/PhysRevA.90.022316} {\bibfield  {journal} {\bibinfo  {journal} {Phys.
  Rev. A}\ }\textbf {\bibinfo {volume} {90}},\ \bibinfo {pages} {022316}
  (\bibinfo {year} {2014})}\BibitemShut {NoStop}%
\bibitem [{\citenamefont {Augusiak}\ \emph {et~al.}(2018)\citenamefont
  {Augusiak}, \citenamefont {Demianowicz},\ and\ \citenamefont
  {Tura}}]{polacos-gme-local}%
  \BibitemOpen
  \bibfield  {author} {\bibinfo {author} {\bibfnamefont {R.}~\bibnamefont
  {Augusiak}}, \bibinfo {author} {\bibfnamefont {M.}~\bibnamefont
  {Demianowicz}}, \ and\ \bibinfo {author} {\bibfnamefont {J.}~\bibnamefont
  {Tura}},\ }\href {\doibase 10.1103/PhysRevA.98.012321} {\bibfield  {journal}
  {\bibinfo  {journal} {Phys. Rev. A}\ }\textbf {\bibinfo {volume} {98}},\
  \bibinfo {pages} {012321} (\bibinfo {year} {2018})}\BibitemShut {NoStop}%
\bibitem [{\citenamefont {Fr\"owis}\ \emph {et~al.}(2017)\citenamefont
  {Fr\"owis}, \citenamefont {Strassman}, \citenamefont {Tiranov}, \citenamefont
  {C.}, \citenamefont {Lavoie}, \citenamefont {Busseries}, \citenamefont
  {Afzelius},\ and\ \citenamefont {Gisin}}]{BrunnerExp2017}%
  \BibitemOpen
  \bibfield  {author} {\bibinfo {author} {\bibfnamefont {F.}~\bibnamefont
  {Fr\"owis}}, \bibinfo {author} {\bibfnamefont {P.~C.}\ \bibnamefont
  {Strassman}}, \bibinfo {author} {\bibfnamefont {A.}~\bibnamefont {Tiranov}},
  \bibinfo {author} {\bibfnamefont {G.}~\bibnamefont {C.}}, \bibinfo {author}
  {\bibfnamefont {N.}~\bibnamefont {Lavoie}, \bibfnamefont {J.~Brunner}},
  \bibinfo {author} {\bibfnamefont {F.}~\bibnamefont {Busseries}}, \bibinfo
  {author} {\bibfnamefont {M.}~\bibnamefont {Afzelius}}, \ and\ \bibinfo
  {author} {\bibfnamefont {N.}~\bibnamefont {Gisin}},\ }\href@noop {}
  {\bibfield  {journal} {\bibinfo  {journal} {Nature Comm.}\ }\textbf {\bibinfo
  {volume} {8}},\ \bibinfo {pages} {907} (\bibinfo {year} {2017})}\BibitemShut
  {NoStop}%
\bibitem [{\citenamefont {Demianowicz}\ and\ \citenamefont
  {Augusiak}(2018)}]{upb-to-ges}%
  \BibitemOpen
  \bibfield  {author} {\bibinfo {author} {\bibfnamefont {M.}~\bibnamefont
  {Demianowicz}}\ and\ \bibinfo {author} {\bibfnamefont {R.}~\bibnamefont
  {Augusiak}},\ }\href {\doibase 10.1103/PhysRevA.98.012313} {\bibfield
  {journal} {\bibinfo  {journal} {Phys. Rev. A}\ }\textbf {\bibinfo {volume}
  {98}},\ \bibinfo {pages} {012313} (\bibinfo {year} {2018})}\BibitemShut
  {NoStop}%
\bibitem [{\citenamefont {Cubitt}\ \emph {et~al.}(2008)\citenamefont {Cubitt},
  \citenamefont {Montanaro},\ and\ \citenamefont {Winter}}]{schmidt-rank}%
  \BibitemOpen
  \bibfield  {author} {\bibinfo {author} {\bibfnamefont {T.}~\bibnamefont
  {Cubitt}}, \bibinfo {author} {\bibfnamefont {A.}~\bibnamefont {Montanaro}}, \
  and\ \bibinfo {author} {\bibfnamefont {A.}~\bibnamefont {Winter}},\ }\href
  {\doibase 10.1063/1.2862998} {\bibfield  {journal} {\bibinfo  {journal} {J.
  Math. Phys.}\ }\textbf {\bibinfo {volume} {49}},\ \bibinfo {pages} {022107}
  (\bibinfo {year} {2008})}\BibitemShut {NoStop}%
\bibitem [{\citenamefont {Parthasarathy}(2004)}]{ces-partha}%
  \BibitemOpen
  \bibfield  {author} {\bibinfo {author} {\bibfnamefont {K.}~\bibnamefont
  {Parthasarathy}},\ }\href {\doibase 10.1007/BF02829441} {\bibfield  {journal}
  {\bibinfo  {journal} {Int. J. Quantum Inform.}\ }\textbf {\bibinfo {volume}
  {114}},\ \bibinfo {pages} {365} (\bibinfo {year} {2004})}\BibitemShut
  {NoStop}%
\bibitem [{\citenamefont {Bhat}(2006)}]{ces-bhat}%
  \BibitemOpen
  \bibfield  {author} {\bibinfo {author} {\bibfnamefont {B.~V.~R.}\
  \bibnamefont {Bhat}},\ }\href {\doibase 10.1142/S0219749906001797} {\bibfield
   {journal} {\bibinfo  {journal} {Int. J. Quantum Inform.}\ }\textbf {\bibinfo
  {volume} {04}},\ \bibinfo {pages} {325} (\bibinfo {year} {2006})}\BibitemShut
  {NoStop}%
\bibitem [{\citenamefont {Gour}\ and\ \citenamefont
  {Wallach}(2007)}]{GourWallach}%
  \BibitemOpen
  \bibfield  {author} {\bibinfo {author} {\bibfnamefont {G.}~\bibnamefont
  {Gour}}\ and\ \bibinfo {author} {\bibfnamefont {N.~R.}\ \bibnamefont
  {Wallach}},\ }\href {\doibase 10.1103/PhysRevA.76.042309} {\bibfield
  {journal} {\bibinfo  {journal} {Phys. Rev. A}\ }\textbf {\bibinfo {volume}
  {76}},\ \bibinfo {pages} {042309} (\bibinfo {year} {2007})}\BibitemShut
  {NoStop}%
\bibitem [{\citenamefont {Raissi}\ \emph {et~al.}(2018)\citenamefont {Raissi},
  \citenamefont {Gogolin}, \citenamefont {Riera},\ and\ \citenamefont
  {Ac{\'{\i}}n}}]{zahra}%
  \BibitemOpen
  \bibfield  {author} {\bibinfo {author} {\bibfnamefont {Z.}~\bibnamefont
  {Raissi}}, \bibinfo {author} {\bibfnamefont {C.}~\bibnamefont {Gogolin}},
  \bibinfo {author} {\bibfnamefont {A.}~\bibnamefont {Riera}}, \ and\ \bibinfo
  {author} {\bibfnamefont {A.}~\bibnamefont {Ac{\'{\i}}n}},\ }\href {\doibase
  10.1088/1751-8121/aaa151} {\bibfield  {journal} {\bibinfo  {journal} {Journal
  of Physics A: Mathematical and Theoretical}\ }\textbf {\bibinfo {volume}
  {51}},\ \bibinfo {pages} {075301} (\bibinfo {year} {2018})}\BibitemShut
  {NoStop}%
\bibitem [{\citenamefont {Ball}(2019)}]{ball}%
  \BibitemOpen
  \bibfield  {author} {\bibinfo {author} {\bibfnamefont {S.}~\bibnamefont
  {Ball}},\ }\href {https://arxiv.org/pdf/1907.04391.pdf} {\bibfield  {journal}
  {\bibinfo  {journal} {arXiv:1907.04391v2 [quant-ph]}\ } (\bibinfo {year}
  {2019})}\BibitemShut {NoStop}%
\bibitem [{\citenamefont {Huber}\ and\ \citenamefont
  {Grassl}(2019)}]{felix-arxiv}%
  \BibitemOpen
  \bibfield  {author} {\bibinfo {author} {\bibfnamefont {F.}~\bibnamefont
  {Huber}}\ and\ \bibinfo {author} {\bibfnamefont {M.}~\bibnamefont {Grassl}},\
  }\href {https://arxiv.org/pdf/1907.07733.pdf} {\bibfield  {journal} {\bibinfo
   {journal} {arXiv:1907.07733 [quant-ph]}\ } (\bibinfo {year}
  {2019})}\BibitemShut {NoStop}%
\bibitem [{\citenamefont {Alsina}\ and\ \citenamefont
  {Razavi}(2019)}]{AME-alsina}%
  \BibitemOpen
  \bibfield  {author} {\bibinfo {author} {\bibfnamefont {D.}~\bibnamefont
  {Alsina}}\ and\ \bibinfo {author} {\bibfnamefont {M.}~\bibnamefont
  {Razavi}},\ }\href {https://arxiv.org/pdf/1907.11253.pdf} {\bibfield
  {journal} {\bibinfo  {journal} {arXiv:1907.11253 [quant-ph]}\ } (\bibinfo
  {year} {2019})}\BibitemShut {NoStop}%
\bibitem [{\citenamefont {Helwig}\ \emph {et~al.}(2012)\citenamefont {Helwig},
  \citenamefont {Cui}, \citenamefont {Latorre}, \citenamefont {Riera},\ and\
  \citenamefont {Lo}}]{Helwig}%
  \BibitemOpen
  \bibfield  {author} {\bibinfo {author} {\bibfnamefont {W.}~\bibnamefont
  {Helwig}}, \bibinfo {author} {\bibfnamefont {W.}~\bibnamefont {Cui}},
  \bibinfo {author} {\bibfnamefont {J.~I.}\ \bibnamefont {Latorre}}, \bibinfo
  {author} {\bibfnamefont {A.}~\bibnamefont {Riera}}, \ and\ \bibinfo {author}
  {\bibfnamefont {H.-K.}\ \bibnamefont {Lo}},\ }\href {\doibase
  10.1103/PhysRevA.86.052335} {\bibfield  {journal} {\bibinfo  {journal} {Phys.
  Rev. A}\ }\textbf {\bibinfo {volume} {86}},\ \bibinfo {pages} {052335}
  (\bibinfo {year} {2012})}\BibitemShut {NoStop}%
\bibitem [{\citenamefont {Goyeneche}\ \emph {et~al.}(2018)\citenamefont
  {Goyeneche}, \citenamefont {Raissi}, \citenamefont {Di~Martino},\ and\
  \citenamefont {\ifmmode~\dot{Z}\else \.{Z}\fi{}yczkowski}}]{dardo-AME}%
  \BibitemOpen
  \bibfield  {author} {\bibinfo {author} {\bibfnamefont {D.}~\bibnamefont
  {Goyeneche}}, \bibinfo {author} {\bibfnamefont {Z.}~\bibnamefont {Raissi}},
  \bibinfo {author} {\bibfnamefont {S.}~\bibnamefont {Di~Martino}}, \ and\
  \bibinfo {author} {\bibfnamefont {K.}~\bibnamefont {\ifmmode~\dot{Z}\else
  \.{Z}\fi{}yczkowski}},\ }\href {\doibase 10.1103/PhysRevA.97.062326}
  {\bibfield  {journal} {\bibinfo  {journal} {Phys. Rev. A}\ }\textbf {\bibinfo
  {volume} {97}},\ \bibinfo {pages} {062326} (\bibinfo {year}
  {2018})}\BibitemShut {NoStop}%
\bibitem [{\citenamefont {Wang}\ \emph {et~al.}(2019)\citenamefont {Wang},
  \citenamefont {Chen}, \citenamefont {Zhao},\ and\ \citenamefont
  {Guo}}]{Wang2019-ges}%
  \BibitemOpen
  \bibfield  {author} {\bibinfo {author} {\bibfnamefont {K.}~\bibnamefont
  {Wang}}, \bibinfo {author} {\bibfnamefont {L.}~\bibnamefont {Chen}}, \bibinfo
  {author} {\bibfnamefont {L.}~\bibnamefont {Zhao}}, \ and\ \bibinfo {author}
  {\bibfnamefont {Y.}~\bibnamefont {Guo}},\ }\href {\doibase
  10.1007/s11128-019-2324-4} {\bibfield  {journal} {\bibinfo  {journal}
  {Quantum Information Processing}\ }\textbf {\bibinfo {volume} {18}},\
  \bibinfo {pages} {202} (\bibinfo {year} {2019})}\BibitemShut {NoStop}%
\bibitem [{\citenamefont {Bennett}\ \emph {et~al.}(1999)\citenamefont
  {Bennett}, \citenamefont {DiVincenzo}, \citenamefont {Mor}, \citenamefont
  {Shor}, \citenamefont {Smolin},\ and\ \citenamefont {Terhal}}]{upb-bennett}%
  \BibitemOpen
  \bibfield  {author} {\bibinfo {author} {\bibfnamefont {C.~H.}\ \bibnamefont
  {Bennett}}, \bibinfo {author} {\bibfnamefont {D.~P.}\ \bibnamefont
  {DiVincenzo}}, \bibinfo {author} {\bibfnamefont {T.}~\bibnamefont {Mor}},
  \bibinfo {author} {\bibfnamefont {P.~W.}\ \bibnamefont {Shor}}, \bibinfo
  {author} {\bibfnamefont {J.~A.}\ \bibnamefont {Smolin}}, \ and\ \bibinfo
  {author} {\bibfnamefont {B.~M.}\ \bibnamefont {Terhal}},\ }\href {\doibase
  10.1103/PhysRevLett.82.5385} {\bibfield  {journal} {\bibinfo  {journal}
  {Phys. Rev. Lett.}\ }\textbf {\bibinfo {volume} {82}},\ \bibinfo {pages}
  {5385} (\bibinfo {year} {1999})}\BibitemShut {NoStop}%
\bibitem [{\citenamefont {De~Rinaldis}(2004)}]{DeRinaldis}%
  \BibitemOpen
  \bibfield  {author} {\bibinfo {author} {\bibfnamefont {S.}~\bibnamefont
  {De~Rinaldis}},\ }\href {\doibase 10.1103/PhysRevA.70.022309} {\bibfield
  {journal} {\bibinfo  {journal} {Phys. Rev. A}\ }\textbf {\bibinfo {volume}
  {70}},\ \bibinfo {pages} {022309} (\bibinfo {year} {2004})}\BibitemShut
  {NoStop}%
\bibitem [{\citenamefont {Duan}\ \emph {et~al.}(2010)\citenamefont {Duan},
  \citenamefont {Xin},\ and\ \citenamefont {Ying}}]{Duan2010}%
  \BibitemOpen
  \bibfield  {author} {\bibinfo {author} {\bibfnamefont {R.}~\bibnamefont
  {Duan}}, \bibinfo {author} {\bibfnamefont {Y.}~\bibnamefont {Xin}}, \ and\
  \bibinfo {author} {\bibfnamefont {M.}~\bibnamefont {Ying}},\ }\href {\doibase
  10.1103/PhysRevA.81.032329} {\bibfield  {journal} {\bibinfo  {journal} {Phys.
  Rev. A}\ }\textbf {\bibinfo {volume} {81}},\ \bibinfo {pages} {032329}
  (\bibinfo {year} {2010})}\BibitemShut {NoStop}%
\bibitem [{\citenamefont {Augusiak}\ \emph {et~al.}(2011)\citenamefont
  {Augusiak}, \citenamefont {Stasi\ifmmode~\acute{n}\else \'{n}\fi{}ska},
  \citenamefont {Hadley}, \citenamefont {Korbicz}, \citenamefont {Lewenstein},\
  and\ \citenamefont {Ac\'{\i}n}}]{upb-prl}%
  \BibitemOpen
  \bibfield  {author} {\bibinfo {author} {\bibfnamefont {R.}~\bibnamefont
  {Augusiak}}, \bibinfo {author} {\bibfnamefont {J.}~\bibnamefont
  {Stasi\ifmmode~\acute{n}\else \'{n}\fi{}ska}}, \bibinfo {author}
  {\bibfnamefont {C.}~\bibnamefont {Hadley}}, \bibinfo {author} {\bibfnamefont
  {J.~K.}\ \bibnamefont {Korbicz}}, \bibinfo {author} {\bibfnamefont
  {M.}~\bibnamefont {Lewenstein}}, \ and\ \bibinfo {author} {\bibfnamefont
  {A.}~\bibnamefont {Ac\'{\i}n}},\ }\href {\doibase
  10.1103/PhysRevLett.107.070401} {\bibfield  {journal} {\bibinfo  {journal}
  {Phys. Rev. Lett.}\ }\textbf {\bibinfo {volume} {107}},\ \bibinfo {pages}
  {070401} (\bibinfo {year} {2011})}\BibitemShut {NoStop}%
\bibitem [{\citenamefont {Agrawal}\ \emph {et~al.}(2019)\citenamefont
  {Agrawal}, \citenamefont {Halder},\ and\ \citenamefont
  {Banik}}]{ManikBanik-ges}%
  \BibitemOpen
  \bibfield  {author} {\bibinfo {author} {\bibfnamefont {S.}~\bibnamefont
  {Agrawal}}, \bibinfo {author} {\bibfnamefont {S.}~\bibnamefont {Halder}}, \
  and\ \bibinfo {author} {\bibfnamefont {M.}~\bibnamefont {Banik}},\ }\href
  {\doibase 10.1103/PhysRevA.99.032335} {\bibfield  {journal} {\bibinfo
  {journal} {Phys. Rev. A}\ }\textbf {\bibinfo {volume} {99}},\ \bibinfo
  {pages} {032335} (\bibinfo {year} {2019})}\BibitemShut {NoStop}%
\bibitem [{\citenamefont {Zhu}(2009)}]{average-ent}%
  \BibitemOpen
  \bibfield  {author} {\bibinfo {author} {\bibfnamefont {G.-Q.}\ \bibnamefont
  {Zhu}},\ }\href {\doibase 10.2478/s11534-008-0129-7} {\bibfield  {journal}
  {\bibinfo  {journal} {Central European Journal of Physics}\ }\textbf
  {\bibinfo {volume} {7}},\ \bibinfo {pages} {135} (\bibinfo {year}
  {2009})}\BibitemShut {NoStop}%
\bibitem [{\citenamefont {Cappellini}\ \emph {et~al.}(2006)\citenamefont
  {Cappellini}, \citenamefont {Sommers},\ and\ \citenamefont
  {\ifmmode~\dot{Z}\else \.{Z}\fi{}yczkowski}}]{average-Gconcurrence}%
  \BibitemOpen
  \bibfield  {author} {\bibinfo {author} {\bibfnamefont {V.}~\bibnamefont
  {Cappellini}}, \bibinfo {author} {\bibfnamefont {H.-J.}\ \bibnamefont
  {Sommers}}, \ and\ \bibinfo {author} {\bibfnamefont {K.}~\bibnamefont
  {\ifmmode~\dot{Z}\else \.{Z}\fi{}yczkowski}},\ }\href {\doibase
  10.1103/PhysRevA.74.062322} {\bibfield  {journal} {\bibinfo  {journal} {Phys.
  Rev. A}\ }\textbf {\bibinfo {volume} {74}},\ \bibinfo {pages} {062322}
  (\bibinfo {year} {2006})}\BibitemShut {NoStop}%
\bibitem [{\citenamefont {Gour}\ and\ \citenamefont
  {Roy}(2008)}]{Gour-max-ent}%
  \BibitemOpen
  \bibfield  {author} {\bibinfo {author} {\bibfnamefont {G.}~\bibnamefont
  {Gour}}\ and\ \bibinfo {author} {\bibfnamefont {A.}~\bibnamefont {Roy}},\
  }\href {\doibase 10.1103/PhysRevA.77.012336} {\bibfield  {journal} {\bibinfo
  {journal} {Phys. Rev. A}\ }\textbf {\bibinfo {volume} {77}},\ \bibinfo
  {pages} {012336} (\bibinfo {year} {2008})}\BibitemShut {NoStop}%
\bibitem [{\citenamefont {Shimony}(1995)}]{Shimony-geometric}%
  \BibitemOpen
  \bibfield  {author} {\bibinfo {author} {\bibfnamefont {A.}~\bibnamefont
  {Shimony}},\ }\href {\doibase 10.1111/j.1749-6632.1995.tb39008.x} {\bibfield
  {journal} {\bibinfo  {journal} {Annals of the New York Academy of Sciences}\
  }\textbf {\bibinfo {volume} {755}},\ \bibinfo {pages} {675} (\bibinfo {year}
  {1995})}\BibitemShut {NoStop}%
\bibitem [{\citenamefont {Barnum}\ and\ \citenamefont
  {Linden}(2001)}]{BL-geometric}%
  \BibitemOpen
  \bibfield  {author} {\bibinfo {author} {\bibfnamefont {H.}~\bibnamefont
  {Barnum}}\ and\ \bibinfo {author} {\bibfnamefont {N.}~\bibnamefont
  {Linden}},\ }\href {http://stacks.iop.org/0305-4470/34/i=35/a=305} {\bibfield
   {journal} {\bibinfo  {journal} {Journal of Physics A: Mathematical and
  General}\ }\textbf {\bibinfo {volume} {34}},\ \bibinfo {pages} {6787}
  (\bibinfo {year} {2001})}\BibitemShut {NoStop}%
\bibitem [{\citenamefont {Wei}\ and\ \citenamefont
  {Goldbart}(2003)}]{WeiGoldbart}%
  \BibitemOpen
  \bibfield  {author} {\bibinfo {author} {\bibfnamefont {T.-C.}\ \bibnamefont
  {Wei}}\ and\ \bibinfo {author} {\bibfnamefont {P.~M.}\ \bibnamefont
  {Goldbart}},\ }\href {\doibase 10.1103/PhysRevA.68.042307} {\bibfield
  {journal} {\bibinfo  {journal} {Phys. Rev. A}\ }\textbf {\bibinfo {volume}
  {68}},\ \bibinfo {pages} {042307} (\bibinfo {year} {2003})}\BibitemShut
  {NoStop}%
\bibitem [{\citenamefont {Das}\ \emph {et~al.}(2016)\citenamefont {Das},
  \citenamefont {Roy}, \citenamefont {Bagchi}, \citenamefont {Misra},
  \citenamefont {Sen(De)},\ and\ \citenamefont {Sen}}]{GGM-multiparty}%
  \BibitemOpen
  \bibfield  {author} {\bibinfo {author} {\bibfnamefont {T.}~\bibnamefont
  {Das}}, \bibinfo {author} {\bibfnamefont {S.~S.}\ \bibnamefont {Roy}},
  \bibinfo {author} {\bibfnamefont {S.}~\bibnamefont {Bagchi}}, \bibinfo
  {author} {\bibfnamefont {A.}~\bibnamefont {Misra}}, \bibinfo {author}
  {\bibfnamefont {A.}~\bibnamefont {Sen(De)}}, \ and\ \bibinfo {author}
  {\bibfnamefont {U.}~\bibnamefont {Sen}},\ }\href {\doibase
  10.1103/PhysRevA.94.022336} {\bibfield  {journal} {\bibinfo  {journal} {Phys.
  Rev. A}\ }\textbf {\bibinfo {volume} {94}},\ \bibinfo {pages} {022336}
  (\bibinfo {year} {2016})}\BibitemShut {NoStop}%
\bibitem [{\citenamefont {Markham}\ \emph {et~al.}(2007)\citenamefont
  {Markham}, \citenamefont {Miyake},\ and\ \citenamefont
  {Virmani}}]{Markham-2007}%
  \BibitemOpen
  \bibfield  {author} {\bibinfo {author} {\bibfnamefont {D.}~\bibnamefont
  {Markham}}, \bibinfo {author} {\bibfnamefont {A.}~\bibnamefont {Miyake}}, \
  and\ \bibinfo {author} {\bibfnamefont {S.}~\bibnamefont {Virmani}},\ }\href
  {\doibase 10.1088/1367-2630/9/6/194} {\bibfield  {journal} {\bibinfo
  {journal} {New Journal of Physics}\ }\textbf {\bibinfo {volume} {9}},\
  \bibinfo {pages} {194} (\bibinfo {year} {2007})}\BibitemShut {NoStop}%
\bibitem [{\citenamefont {Or\'us}(2008)}]{orus-2008}%
  \BibitemOpen
  \bibfield  {author} {\bibinfo {author} {\bibfnamefont {R.}~\bibnamefont
  {Or\'us}},\ }\href {\doibase 10.1103/PhysRevLett.100.130502} {\bibfield
  {journal} {\bibinfo  {journal} {Phys. Rev. Lett.}\ }\textbf {\bibinfo
  {volume} {100}},\ \bibinfo {pages} {130502} (\bibinfo {year}
  {2008})}\BibitemShut {NoStop}%
\bibitem [{\citenamefont {Hayashi}\ \emph {et~al.}(2008)\citenamefont
  {Hayashi}, \citenamefont {Markham}, \citenamefont {Murao}, \citenamefont
  {Owari},\ and\ \citenamefont {Virmani}}]{hayashi-gm-application}%
  \BibitemOpen
  \bibfield  {author} {\bibinfo {author} {\bibfnamefont {M.}~\bibnamefont
  {Hayashi}}, \bibinfo {author} {\bibfnamefont {D.}~\bibnamefont {Markham}},
  \bibinfo {author} {\bibfnamefont {M.}~\bibnamefont {Murao}}, \bibinfo
  {author} {\bibfnamefont {M.}~\bibnamefont {Owari}}, \ and\ \bibinfo {author}
  {\bibfnamefont {S.}~\bibnamefont {Virmani}},\ }\href {\doibase
  10.1103/PhysRevA.77.012104} {\bibfield  {journal} {\bibinfo  {journal} {Phys.
  Rev. A}\ }\textbf {\bibinfo {volume} {77}},\ \bibinfo {pages} {012104}
  (\bibinfo {year} {2008})}\BibitemShut {NoStop}%
\bibitem [{\citenamefont {Gross}\ \emph {et~al.}(2009)\citenamefont {Gross},
  \citenamefont {Flammia},\ and\ \citenamefont {Eisert}}]{gross-to-entangled}%
  \BibitemOpen
  \bibfield  {author} {\bibinfo {author} {\bibfnamefont {D.}~\bibnamefont
  {Gross}}, \bibinfo {author} {\bibfnamefont {S.~T.}\ \bibnamefont {Flammia}},
  \ and\ \bibinfo {author} {\bibfnamefont {J.}~\bibnamefont {Eisert}},\ }\href
  {\doibase 10.1103/PhysRevLett.102.190501} {\bibfield  {journal} {\bibinfo
  {journal} {Phys. Rev. Lett.}\ }\textbf {\bibinfo {volume} {102}},\ \bibinfo
  {pages} {190501} (\bibinfo {year} {2009})}\BibitemShut {NoStop}%
\bibitem [{\citenamefont {Kaszlikowski}\ \emph {et~al.}(2008)\citenamefont
  {Kaszlikowski}, \citenamefont {Sen(De)}, \citenamefont {Sen}, \citenamefont
  {Vedral},\ and\ \citenamefont {Winter}}]{WW-GES}%
  \BibitemOpen
  \bibfield  {author} {\bibinfo {author} {\bibfnamefont {D.}~\bibnamefont
  {Kaszlikowski}}, \bibinfo {author} {\bibfnamefont {A.}~\bibnamefont
  {Sen(De)}}, \bibinfo {author} {\bibfnamefont {U.}~\bibnamefont {Sen}},
  \bibinfo {author} {\bibfnamefont {V.}~\bibnamefont {Vedral}}, \ and\ \bibinfo
  {author} {\bibfnamefont {A.}~\bibnamefont {Winter}},\ }\href {\doibase
  10.1103/PhysRevLett.101.070502} {\bibfield  {journal} {\bibinfo  {journal}
  {Phys. Rev. Lett.}\ }\textbf {\bibinfo {volume} {101}},\ \bibinfo {pages}
  {070502} (\bibinfo {year} {2008})}\BibitemShut {NoStop}%
\bibitem [{\citenamefont {Sen(De)}\ and\ \citenamefont
  {Sen}(2010)}]{Senowie-GGMoE}%
  \BibitemOpen
  \bibfield  {author} {\bibinfo {author} {\bibfnamefont {A.}~\bibnamefont
  {Sen(De)}}\ and\ \bibinfo {author} {\bibfnamefont {U.}~\bibnamefont {Sen}},\
  }\href {\doibase 10.1103/PhysRevA.81.012308} {\bibfield  {journal} {\bibinfo
  {journal} {Phys. Rev. A}\ }\textbf {\bibinfo {volume} {81}},\ \bibinfo
  {pages} {012308} (\bibinfo {year} {2010})}\BibitemShut {NoStop}%
\bibitem [{\citenamefont {Branciard}\ \emph {et~al.}(2010)\citenamefont
  {Branciard}, \citenamefont {Zhu}, \citenamefont {Chen},\ and\ \citenamefont
  {Scarani}}]{Branciard}%
  \BibitemOpen
  \bibfield  {author} {\bibinfo {author} {\bibfnamefont {C.}~\bibnamefont
  {Branciard}}, \bibinfo {author} {\bibfnamefont {H.}~\bibnamefont {Zhu}},
  \bibinfo {author} {\bibfnamefont {L.}~\bibnamefont {Chen}}, \ and\ \bibinfo
  {author} {\bibfnamefont {V.}~\bibnamefont {Scarani}},\ }\href {\doibase
  10.1103/PhysRevA.82.012327} {\bibfield  {journal} {\bibinfo  {journal} {Phys.
  Rev. A}\ }\textbf {\bibinfo {volume} {82}},\ \bibinfo {pages} {012327}
  (\bibinfo {year} {2010})}\BibitemShut {NoStop}%
\bibitem [{\citenamefont {Blasone}\ \emph {et~al.}(2008)\citenamefont
  {Blasone}, \citenamefont {Dell'Anno}, \citenamefont {De~Siena},\ and\
  \citenamefont {Illuminati}}]{Blasone2008}%
  \BibitemOpen
  \bibfield  {author} {\bibinfo {author} {\bibfnamefont {M.}~\bibnamefont
  {Blasone}}, \bibinfo {author} {\bibfnamefont {F.}~\bibnamefont {Dell'Anno}},
  \bibinfo {author} {\bibfnamefont {S.}~\bibnamefont {De~Siena}}, \ and\
  \bibinfo {author} {\bibfnamefont {F.}~\bibnamefont {Illuminati}},\ }\href
  {\doibase 10.1103/PhysRevA.77.062304} {\bibfield  {journal} {\bibinfo
  {journal} {Phys. Rev. A}\ }\textbf {\bibinfo {volume} {77}},\ \bibinfo
  {pages} {062304} (\bibinfo {year} {2008})}\BibitemShut {NoStop}%
\bibitem [{\citenamefont {Losonczi}(1992)}]{losonczi}%
  \BibitemOpen
  \bibfield  {author} {\bibinfo {author} {\bibfnamefont {L.}~\bibnamefont
  {Losonczi}},\ }\href@noop {} {\bibfield  {journal} {\bibinfo  {journal} {Acta
  Mathematica Hungarica}\ }\textbf {\bibinfo {volume} {60}},\ \bibinfo {pages}
  {309} (\bibinfo {year} {1992})}\BibitemShut {NoStop}%
\bibitem [{\citenamefont {Huang}(2014)}]{Huang-NP}%
  \BibitemOpen
  \bibfield  {author} {\bibinfo {author} {\bibfnamefont {Y.}~\bibnamefont
  {Huang}},\ }\href {\doibase 10.1088/1367-2630/16/3/033027} {\bibfield
  {journal} {\bibinfo  {journal} {New Journal of Physics}\ }\textbf {\bibinfo
  {volume} {16}},\ \bibinfo {pages} {033027} (\bibinfo {year}
  {2014})}\BibitemShut {NoStop}%
\bibitem [{\citenamefont {Zhang}\ \emph {et~al.}(2019)\citenamefont {Zhang},
  \citenamefont {Dai}, \citenamefont {Dong},\ and\ \citenamefont
  {Zhang}}]{GM-PPT-new}%
  \BibitemOpen
  \bibfield  {author} {\bibinfo {author} {\bibfnamefont {Z.}~\bibnamefont
  {Zhang}}, \bibinfo {author} {\bibfnamefont {Y.}~\bibnamefont {Dai}}, \bibinfo
  {author} {\bibfnamefont {Y.}~\bibnamefont {Dong}}, \ and\ \bibinfo {author}
  {\bibfnamefont {C.}~\bibnamefont {Zhang}},\ }\href
  {https://arxiv.org/pdf/1903.10944.pdf} {\bibfield  {journal} {\bibinfo
  {journal} {arXiv:1903.10944 [quant-ph]}\ } (\bibinfo {year}
  {2019})}\BibitemShut {NoStop}%
\bibitem [{\citenamefont {Aloy}\ \emph {et~al.}(2019)\citenamefont {Aloy},
  \citenamefont {Tura}, \citenamefont {Baccari}, \citenamefont {Acín},
  \citenamefont {Lewenstein},\ and\ \citenamefont {Augusiak}}]{ent-depth}%
  \BibitemOpen
  \bibfield  {author} {\bibinfo {author} {\bibfnamefont {A.}~\bibnamefont
  {Aloy}}, \bibinfo {author} {\bibfnamefont {J.}~\bibnamefont {Tura}}, \bibinfo
  {author} {\bibfnamefont {F.}~\bibnamefont {Baccari}}, \bibinfo {author}
  {\bibfnamefont {A.}~\bibnamefont {Acín}}, \bibinfo {author} {\bibfnamefont
  {M.}~\bibnamefont {Lewenstein}}, \ and\ \bibinfo {author} {\bibfnamefont
  {R.}~\bibnamefont {Augusiak}},\ }\href {https://arxiv.org/pdf/1807.06027.pdf}
  {\bibfield  {journal} {\bibinfo  {journal} {arXiv:1807.06027 [quant-ph]}\ }
  (\bibinfo {year} {2019})}\BibitemShut {NoStop}%
\bibitem [{\citenamefont {Doherty}\ \emph {et~al.}(2004)\citenamefont
  {Doherty}, \citenamefont {Parrilo},\ and\ \citenamefont {Spedalieri}}]{DPS}%
  \BibitemOpen
  \bibfield  {author} {\bibinfo {author} {\bibfnamefont {A.~C.}\ \bibnamefont
  {Doherty}}, \bibinfo {author} {\bibfnamefont {P.~A.}\ \bibnamefont
  {Parrilo}}, \ and\ \bibinfo {author} {\bibfnamefont {F.~M.}\ \bibnamefont
  {Spedalieri}},\ }\href {\doibase 10.1103/PhysRevA.69.022308} {\bibfield
  {journal} {\bibinfo  {journal} {Phys. Rev. A}\ }\textbf {\bibinfo {volume}
  {69}},\ \bibinfo {pages} {022308} (\bibinfo {year} {2004})}\BibitemShut
  {NoStop}%
\bibitem [{\citenamefont {Navascu{\'e}s}\ \emph {et~al.}(2009)\citenamefont
  {Navascu{\'e}s}, \citenamefont {Owari},\ and\ \citenamefont
  {Plenio}}]{power-of-PPT}%
  \BibitemOpen
  \bibfield  {author} {\bibinfo {author} {\bibfnamefont {M.}~\bibnamefont
  {Navascu{\'e}s}}, \bibinfo {author} {\bibfnamefont {M.}~\bibnamefont
  {Owari}}, \ and\ \bibinfo {author} {\bibfnamefont {M.~B.}\ \bibnamefont
  {Plenio}},\ }\enquote {\bibinfo {title} {On the power of the ppt constraint
  in the symmetric extensions test for separability},}\ in\ \href {\doibase
  10.1007/978-3-642-10698-9_10} {\emph {\bibinfo {booktitle} {Theory of Quantum
  Computation, Communication, and Cryptography: 4th Workshop,TQC 2009,
  Waterloo, Canada, May 11-13, 2009, Revised Selected Papers}}},\ \bibinfo
  {editor} {edited by\ \bibinfo {editor} {\bibfnamefont {A.}~\bibnamefont
  {Childs}}\ and\ \bibinfo {editor} {\bibfnamefont {M.}~\bibnamefont {Mosca}}}\
  (\bibinfo  {publisher} {Springer Berlin Heidelberg},\ \bibinfo {address}
  {Berlin, Heidelberg},\ \bibinfo {year} {2009})\ pp.\ \bibinfo {pages}
  {94--106}\BibitemShut {NoStop}%
\bibitem [{\citenamefont {Eisert}\ \emph {et~al.}(2004)\citenamefont {Eisert},
  \citenamefont {Hyllus}, \citenamefont {G\"uhne},\ and\ \citenamefont
  {Curty}}]{EntCriteria-hierarchy}%
  \BibitemOpen
  \bibfield  {author} {\bibinfo {author} {\bibfnamefont {J.}~\bibnamefont
  {Eisert}}, \bibinfo {author} {\bibfnamefont {P.}~\bibnamefont {Hyllus}},
  \bibinfo {author} {\bibfnamefont {O.}~\bibnamefont {G\"uhne}}, \ and\
  \bibinfo {author} {\bibfnamefont {M.}~\bibnamefont {Curty}},\ }\href
  {\doibase 10.1103/PhysRevA.70.062317} {\bibfield  {journal} {\bibinfo
  {journal} {Phys. Rev. A}\ }\textbf {\bibinfo {volume} {70}},\ \bibinfo
  {pages} {062317} (\bibinfo {year} {2004})}\BibitemShut {NoStop}%
\bibitem [{\citenamefont {Breuer}(2006)}]{Breuer}%
  \BibitemOpen
  \bibfield  {author} {\bibinfo {author} {\bibfnamefont {H.-P.}\ \bibnamefont
  {Breuer}},\ }\href {\doibase 10.1103/PhysRevLett.97.080501} {\bibfield
  {journal} {\bibinfo  {journal} {Phys. Rev. Lett.}\ }\textbf {\bibinfo
  {volume} {97}},\ \bibinfo {pages} {080501} (\bibinfo {year}
  {2006})}\BibitemShut {NoStop}%
\bibitem [{\citenamefont {Hall}(2006)}]{hall}%
  \BibitemOpen
  \bibfield  {author} {\bibinfo {author} {\bibfnamefont {W.}~\bibnamefont
  {Hall}},\ }\href {\doibase 10.1088/0305-4470/39/45/020} {\bibfield  {journal}
  {\bibinfo  {journal} {Journal of Physics A: Mathematical and General}\
  }\textbf {\bibinfo {volume} {39}},\ \bibinfo {pages} {14119} (\bibinfo {year}
  {2006})}\BibitemShut {NoStop}%
\bibitem [{\citenamefont {Lancien}\ \emph {et~al.}(2015)\citenamefont
  {Lancien}, \citenamefont {Gühne}, \citenamefont {Sengupta},\ and\
  \citenamefont {Huber}}]{Lancien_2015}%
  \BibitemOpen
  \bibfield  {author} {\bibinfo {author} {\bibfnamefont {C.}~\bibnamefont
  {Lancien}}, \bibinfo {author} {\bibfnamefont {O.}~\bibnamefont {Gühne}},
  \bibinfo {author} {\bibfnamefont {R.}~\bibnamefont {Sengupta}}, \ and\
  \bibinfo {author} {\bibfnamefont {M.}~\bibnamefont {Huber}},\ }\href
  {\doibase 10.1088/1751-8113/48/50/505302} {\bibfield  {journal} {\bibinfo
  {journal} {Journal of Physics A: Mathematical and Theoretical}\ }\textbf
  {\bibinfo {volume} {48}},\ \bibinfo {pages} {505302} (\bibinfo {year}
  {2015})}\BibitemShut {NoStop}%
\bibitem [{\citenamefont {Zhu}\ \emph {et~al.}(2010)\citenamefont {Zhu},
  \citenamefont {Chen},\ and\ \citenamefont {Hayashi}}]{geometric-hayashi}%
  \BibitemOpen
  \bibfield  {author} {\bibinfo {author} {\bibfnamefont {H.}~\bibnamefont
  {Zhu}}, \bibinfo {author} {\bibfnamefont {L.}~\bibnamefont {Chen}}, \ and\
  \bibinfo {author} {\bibfnamefont {M.}~\bibnamefont {Hayashi}},\ }\href
  {\doibase 10.1088/1367-2630/12/8/083002} {\bibfield  {journal} {\bibinfo
  {journal} {New Journal of Physics}\ }\textbf {\bibinfo {volume} {12}},\
  \bibinfo {pages} {083002} (\bibinfo {year} {2010})}\BibitemShut {NoStop}%
\bibitem [{\citenamefont {Bracken}(2004)}]{Bracken}%
  \BibitemOpen
  \bibfield  {author} {\bibinfo {author} {\bibfnamefont {A.~J.}\ \bibnamefont
  {Bracken}},\ }\href {\doibase 10.1103/PhysRevA.69.052331} {\bibfield
  {journal} {\bibinfo  {journal} {Phys. Rev. A}\ }\textbf {\bibinfo {volume}
  {69}},\ \bibinfo {pages} {052331} (\bibinfo {year} {2004})}\BibitemShut
  {NoStop}%
\bibitem [{\citenamefont {Lewenstein}\ \emph {et~al.}(2001)\citenamefont
  {Lewenstein}, \citenamefont {Kraus}, \citenamefont {Horodecki},\ and\
  \citenamefont {Cirac}}]{LewensteinWitnesses2001}%
  \BibitemOpen
  \bibfield  {author} {\bibinfo {author} {\bibfnamefont {M.}~\bibnamefont
  {Lewenstein}}, \bibinfo {author} {\bibfnamefont {B.}~\bibnamefont {Kraus}},
  \bibinfo {author} {\bibfnamefont {P.}~\bibnamefont {Horodecki}}, \ and\
  \bibinfo {author} {\bibfnamefont {J.~I.}\ \bibnamefont {Cirac}},\ }\href
  {\doibase 10.1103/PhysRevA.63.044304} {\bibfield  {journal} {\bibinfo
  {journal} {Phys. Rev. A}\ }\textbf {\bibinfo {volume} {63}},\ \bibinfo
  {pages} {044304} (\bibinfo {year} {2001})}\BibitemShut {NoStop}%
\bibitem [{\citenamefont {Jungnitsch}\ \emph {et~al.}(2011)\citenamefont
  {Jungnitsch}, \citenamefont {Moroder},\ and\ \citenamefont
  {G\"uhne}}]{taming}%
  \BibitemOpen
  \bibfield  {author} {\bibinfo {author} {\bibfnamefont {B.}~\bibnamefont
  {Jungnitsch}}, \bibinfo {author} {\bibfnamefont {T.}~\bibnamefont {Moroder}},
  \ and\ \bibinfo {author} {\bibfnamefont {O.}~\bibnamefont {G\"uhne}},\ }\href
  {\doibase 10.1103/PhysRevLett.106.190502} {\bibfield  {journal} {\bibinfo
  {journal} {Phys. Rev. Lett.}\ }\textbf {\bibinfo {volume} {106}},\ \bibinfo
  {pages} {190502} (\bibinfo {year} {2011})}\BibitemShut {NoStop}%
\bibitem [{\citenamefont {Streltsov}\ \emph {et~al.}(2010)\citenamefont
  {Streltsov}, \citenamefont {Kampermann},\ and\ \citenamefont
  {Bru{\ss}}}]{Streltsov-gm-fidelity}%
  \BibitemOpen
  \bibfield  {author} {\bibinfo {author} {\bibfnamefont {A.}~\bibnamefont
  {Streltsov}}, \bibinfo {author} {\bibfnamefont {H.}~\bibnamefont
  {Kampermann}}, \ and\ \bibinfo {author} {\bibfnamefont {D.}~\bibnamefont
  {Bru{\ss}}},\ }\href {\doibase 10.1088/1367-2630/12/12/123004} {\bibfield
  {journal} {\bibinfo  {journal} {New Journal of Physics}\ }\textbf {\bibinfo
  {volume} {12}},\ \bibinfo {pages} {123004} (\bibinfo {year}
  {2010})}\BibitemShut {NoStop}%
\bibitem [{\citenamefont {Regula}\ \emph {et~al.}(2018)\citenamefont {Regula},
  \citenamefont {Piani}, \citenamefont {Cianciaruso}, \citenamefont {Bromley},
  \citenamefont {Streltsov},\ and\ \citenamefont {Adesso}}]{Regula2018}%
  \BibitemOpen
  \bibfield  {author} {\bibinfo {author} {\bibfnamefont {B.}~\bibnamefont
  {Regula}}, \bibinfo {author} {\bibfnamefont {M.}~\bibnamefont {Piani}},
  \bibinfo {author} {\bibfnamefont {M.}~\bibnamefont {Cianciaruso}}, \bibinfo
  {author} {\bibfnamefont {T.~R.}\ \bibnamefont {Bromley}}, \bibinfo {author}
  {\bibfnamefont {A.}~\bibnamefont {Streltsov}}, \ and\ \bibinfo {author}
  {\bibfnamefont {G.}~\bibnamefont {Adesso}},\ }\href {\doibase
  10.1088/1367-2630/aaae9d} {\bibfield  {journal} {\bibinfo  {journal} {New
  Journal of Physics}\ }\textbf {\bibinfo {volume} {20}},\ \bibinfo {pages}
  {033012} (\bibinfo {year} {2018})}\BibitemShut {NoStop}%
\bibitem [{\citenamefont {Watrous}(2013)}]{Watrous}%
  \BibitemOpen
  \bibfield  {author} {\bibinfo {author} {\bibfnamefont {J.}~\bibnamefont
  {Watrous}},\ }\href@noop {} {\bibfield  {journal} {\bibinfo  {journal}
  {Chicago Journal of Theoretical Computer Science}\ }\textbf {\bibinfo
  {volume} {2013}},\ \bibinfo {pages} {8} (\bibinfo {year} {2013})}\BibitemShut
  {NoStop}%
\bibitem [{\citenamefont {Streltsov}\ \emph {et~al.}(2011)\citenamefont
  {Streltsov}, \citenamefont {Kampermann},\ and\ \citenamefont
  {Bru\ss{}}}]{algorithm-gm}%
  \BibitemOpen
  \bibfield  {author} {\bibinfo {author} {\bibfnamefont {A.}~\bibnamefont
  {Streltsov}}, \bibinfo {author} {\bibfnamefont {H.}~\bibnamefont
  {Kampermann}}, \ and\ \bibinfo {author} {\bibfnamefont {D.}~\bibnamefont
  {Bru\ss{}}},\ }\href {\doibase 10.1103/PhysRevA.84.022323} {\bibfield
  {journal} {\bibinfo  {journal} {Phys. Rev. A}\ }\textbf {\bibinfo {volume}
  {84}},\ \bibinfo {pages} {022323} (\bibinfo {year} {2011})}\BibitemShut
  {NoStop}%
\bibitem [{\citenamefont {Demianowicz}()}]{miniatura}%
  \BibitemOpen
  \bibfield  {author} {\bibinfo {author} {\bibfnamefont {M.}~\bibnamefont
  {Demianowicz}},\ }\href@noop {} {\enquote {\bibinfo {title} {Entanglement
  robustness to local noise of genuinely entangled subspaces},}\ }\bibinfo
  {note} {In preparation}\BibitemShut {NoStop}%
\bibitem [{\citenamefont {Grant}\ and\ \citenamefont {Boyd}(2014)}]{cvx}%
  \BibitemOpen
  \bibfield  {author} {\bibinfo {author} {\bibfnamefont {M.}~\bibnamefont
  {Grant}}\ and\ \bibinfo {author} {\bibfnamefont {S.}~\bibnamefont {Boyd}},\
  }\href@noop {} {\enquote {\bibinfo {title} {{CVX}: Matlab software for
  disciplined convex programming, version 2.1},}\ }\bibinfo {howpublished}
  {\url{http://cvxr.com/cvx}} (\bibinfo {year} {2014})\BibitemShut {NoStop}%
\bibitem [{\citenamefont {Grant}\ and\ \citenamefont
  {Boyd}(2008)}]{GrantBoyd08}%
  \BibitemOpen
  \bibfield  {author} {\bibinfo {author} {\bibfnamefont {M.}~\bibnamefont
  {Grant}}\ and\ \bibinfo {author} {\bibfnamefont {S.}~\bibnamefont {Boyd}},\
  }in\ \href@noop {} {\emph {\bibinfo {booktitle} {Recent Advances in Learning
  and Control}}},\ \bibinfo {series and number} {Lecture Notes in Control and
  Information Sciences},\ \bibinfo {editor} {edited by\ \bibinfo {editor}
  {\bibfnamefont {V.}~\bibnamefont {Blondel}}, \bibinfo {editor} {\bibfnamefont
  {S.}~\bibnamefont {Boyd}}, \ and\ \bibinfo {editor} {\bibfnamefont
  {H.}~\bibnamefont {Kimura}}}\ (\bibinfo  {publisher} {Springer-Verlag
  Limited},\ \bibinfo {year} {2008})\ pp.\ \bibinfo {pages} {95--110},\
  \bibinfo {note} {\url{http://stanford.edu/~boyd/graph_dcp.html}}\BibitemShut
  {NoStop}%
\bibitem [{\citenamefont {Johnston}(2016)}]{qetlab}%
  \BibitemOpen
  \bibfield  {author} {\bibinfo {author} {\bibfnamefont {N.}~\bibnamefont
  {Johnston}},\ }\href {\doibase 10.5281/zenodo.44637} {\enquote {\bibinfo
  {title} {{QETLAB}: A {MATLAB} toolbox for quantum entanglement, version
  0.9},}\ }\bibinfo {howpublished} {\url{http://qetlab.com}} (\bibinfo {year}
  {2016})\BibitemShut {NoStop}%
\end{thebibliography}%
\appendix
\section{Proof of Theorem \ref{ces-ent} (section \ref{ces-new})}\label{app-proof-1}

\begin{proof}
	We construct the matrix $\frakS_A$ from (\ref{S-matrix-A}) choosing the vectors on subsystem $A$ to be $\ket{x}=x_0 \ket{0}+x_1 \ket{1}$ with $|x_0|^2+|x_1|^2=1$. We obtain (we explicitly state the dependence on $x_0$ of this matrix):
	\beqn
	\frakS_A (x_0)&=&\sum_{i=0}^{d-2} \langle{x}\proj{\phi_i}x\rangle \\
	&=& \sum_{i=0}^{d-2} \Bigg( |a x_0|^2 \proj{\psi_i}+ |bx_1|^2 \proj{\psi_{i+1}} 
	\non 
	&&+
	 ax_1(bx_0)^* \ket{\psi_i} \bra{\psi_{i+1}} +(ax_1)^* b x_0  \ket{\psi_{i+1}} \bra{\psi_{i}} \Bigg). \nonumber
	\eeqn
	In the basis of $\ket{\psi_i}$'s it is a $d \times d$ tridiagonal matrix of the form:
	\beqn 
	\frakS_A (x_0)= \left( \begin{array}{cccccc}
		\alpha&g &0 & \cdots &0 &0\\
		g^* & 	\alpha+\beta&g& \cdots &0&0\\
		0& g^* &\alpha+\beta & \cdots &0&0\\
		\vdots & \vdots & \ddots & \ddots &\vdots & \vdots\\
		0& 0  &\ddots &\alpha+\beta&g &0\\
		0& 0 &\cdots & g^*  &\alpha+\beta &g\\
		0&0 &\cdots & 0&g^*  &\beta
	\end{array}     \right),\non
	\eeqn
	where $\alpha= |a x_0|^2$ , $\beta= |b x_1|^2$, $g=  ax_1 (bx_0)^*$.
	
	Crucially, $\alpha \beta= |g|^2$, and using the results of \cite{losonczi} we find the eigenvalues to be:
	\beqn \label{w-wlasne-app}
	\lambda_k \left(\frakS_A(x_0)\right)&=& \alpha+\beta+2|g| \cos \frac{k \pi}{d}\\
	&=&  |a x_0|^2 + |b x_1|^2 + 2 |x_0x_1ab| \cos \frac{k \pi}{d} \nonumber
	\eeqn
	for $k=1,2,\dots, d-1$ and $\lambda_d=0$. 	

	We now need to find 
	\beqn
	\lambda_{\mathrm{max}}(\frakS_A):=\max_{x_0,k} \lambda_k (\frakS_A (x_0)), \quad x_0 \in [0,1].
	\eeqn

	First, let us consider the values on the boundary of the interval for $x_0$. We have that $\lambda(\frakS(x_0=0))=|b|^2=\sin^2(\theta/2)$ and 
	$\lambda(\frakS(x_0=1))=|a|^2=\cos^2(\theta/2)$.
	
	For $x_0\in (0,1)$ it is clear that we can restrict ourselves to the eigenvalues (\ref{w-wlasne}) for which the cosine is larger than zero.
	For a given such $k$ we have (recall that  $|x_0|^2+|x_1|^2=1$):
	\beqn
	\frac{\mathrm{d} \lambda_k (\frakS_A (x_0))}{\mathrm{d} |x_0|} &=& \\ && \hspace{-2cm} 2 (|a|^2-|b|^2) |x_0|+2|ab|\cos \left(\frac{k \pi}{d}\right) \frac{1-2|x_0|^2}{\sqrt{1-|x_0|^2}}.  \nonumber
	\eeqn
	Denoting $p:=|a|^2-|b|^2$ and $q_k:=|ab|\cos \left(\frac{k \pi}{d}\right)>0$ we equate this to zero to obtain:
	\beqn \label{zaleznosc}
	p|x_0|(1-|x_0|^2)=q_k(2|x_0|^2-1).
	\eeqn
	It follows that 
	\beqn\label{p}
	&&p>0\; (|a|>|b|) \implies |x_0|>\frac{1}{\sqrt{2}},\\
	&&p<0\; (|a|<|b|) \implies |x_0|<\frac{1}{\sqrt{2}}. \label{p2}
	\eeqn
	When $p=0$, i.e., $|a|=|b|$ we have $|x_0|=1/\sqrt{2}$. Squaring (\ref{zaleznosc}) we obtain a biquadratic equation on $|x_0|$:
	\beqn
	(p^2+4q_k^2)|x_0|^4-(p^2+4q_k^2)|x_0|^2+q_k^2=0.
	\eeqn
	Having in mind (\ref{p}) and (\ref{p2}) we then conclude that the optimal $x_0=\tilde{x}_0$ is given by:
	\beqn \displaystyle
	\tilde{x}_0=\eksp^{\uroj \xi}\sqrt{\frac{1}{2} \left( 1 +  \frac{p}{\sqrt{p^2+4q_k^2}} \right)}, \quad \xi \in \mathbb{R}
	\eeqn
	with the denotation as above. 
	
	This results in:
	\beqn
	\lambda_k (\frakS_A(\tilde{x}_0))= \frac{1}{2} \left( 1+\sqrt{p^2+4q_k^2}  \right)
	\eeqn
	and in turn
	\beqn
	\lambda_{\mathrm{max}}(\frakS_A)&=&\lambda_1 (\frakS_A(\tilde{x}_0))\\
	&=&\frac{1}{2} \left( 1+\sqrt{(|a|^2-|b|^2)^2+4|ab|^2 \cos^2 \left(\frac{ \pi}{d}  \right) }\right) \non
	&=& \frac{1}{2} \left( 1+
	\sqrt{1-\sin^2\theta \sin^2 \left(\frac{ \pi}{d}\right) }   \right).
	\nonumber
	\eeqn
	This is larger than both $\sin^2(\theta/2)$ and $\cos^2(\theta/2)$ obtained on the boundaries and
	it follows that the entanglement of $\calS_{2\times d}^{\theta}$ is given by
	\beqn
	\hspace{-0.5cm}
	E_{GM}(\calS_{2\times d}^{\theta})=\frac{1}{2} \left( 1-
	\sqrt{1-\sin^2\theta \sin^2 \left(\frac{ \pi}{d}\right) }   \right).
	\eeqn
	The nonzero value of the entanglement signifies that the subspace is indeed a CES.
	
	in the case $\theta=\pi/2$, i.e., $|a|=|b|=1/\sqrt{2}$,  we obviously have:
	\beqn
	E_{GM}(\calS_{2\times d}^{\pi/2})=\frac{1}{2}\left(1-\cos \frac{\pi}{d}\right)=\sin^2  \frac{\pi}{2d}.
	\eeqn
\end{proof}

\section{The case of three parties: reduction of the GES from Theorem 1 of \cite{upb-to-ges} to $\calS_{2\times d^2}^{\pi/2}$}\label{app-th1}
Here we consider a class of GESs considered recently in Theorem 1 of \cite{upb-to-ges}. It is defined as being orthogonal to the nonorthogonal unextendible product basis given by the following vectors ($\alpha \in \mathbb{C}$):
\beqn\label{vec-upb}
&&(1,\alpha^{\tilde{d}},\alpha^{2\tilde{d}},\dots, \alpha^{(d-1)\tilde{d}})_{A_1}\otimes (1,\alpha,\alpha^2,\dots,\alpha^{d^{N-1}-1})_{\aaa{2}{k}}\non
&&\hspace{+0.5cm}= \left(\sum_{i=0}^{d-1} \alpha^{i\tilde{d}} \ket{i}_{A_1}\right)\otimes  \left(\sum_{k=0}^{d^{N-1}-1} \alpha^k \ket{k}_{\aaa{2}{k}} \right) \non
&&\hspace{+0.5cm}= \sum_{i=0}^{d-1}\sum_{k=0}^{d^{N-1}-1} \alpha^{i\tilde{d}+k} \ket{i}_{A_1}\ket{k}_{\aaa{2}{k}} 
\eeqn
with
%
$\tilde{d} = d^{N-1}-d+1$.
%
%

In \cite{upb-to-ges} we have not given an explicit basis for the resulting GES but only showed that its dimension is $(d-1)^2$. We now provide such basis and provide an explanation on why the dimension does not scale with the number of parties $N$. 

By looking at the repeating monomials of $\alpha$ in the coordinates of the vectors (\ref{vec-upb}) after performing tensor multiplication, one can easily verifies that, regardless of the number of parties $N$, the vectors orthogonal to the vectors (\ref{vec-upb}), i.e., the vectors spanning the corresponding GES are:
\beqn \label{th1-ges}
&&\frac{1}{\sqrt{2}}\left(\ket{i}_{A_1}\ket{\tylda+k}_{\aaa{2}{k}}-\ket{i+1}_{A_1}\ket{k}_{\aaa{2}{k}}\right), \non&& \hspace{+0.5cm} i=0,1,\dots,d-2, \quad
k=0,1,\dots,d-2.
\eeqn
One sees that they form an orthonormal set.

Importantly, on qudits $A_2,\dots, A_{N-1}$ one has in fact a qubit subspace. This is because:
\beqn
&&\ket{\tylda+k}_{\aaa{2}{k}}=\ket{d-1}_{A_2}\cdots \ket{d-1}_{A_{N-1}}\ket{k+1}_{A_N},\non
&&\ket{k}_{\aaa{2}{k}}=\ket{0}_{A_2}\ket{0}_{A_3}\cdots \ket{0}_{A_{N-1}}\ket{k}_{A_N}.
\eeqn
Denoting further
\beqn
&&\ket{\bar{1}}:=\ket{d-1}^{\otimes (N-2)},\quad \ket{\bar{0}}:=\ket{0}^{\otimes (N-2)}.
\eeqn
and substituting $A_1 \rightarrow A$, $A_2\dots A_{N-1} \rightarrow B$, $A_N\rightarrow C$, we can write down the vectors from (\ref{th1-ges}) as follows:
\beqn \label{ges-th1-vecs}
&&\frac{1}{\sqrt{2}}\left(\ket{i}_A \ket{\bar{1}}_B\ket{k+1}_C-\ket{i+1}_A\ket{\bar{0}}_B\ket{k}_C\right), \non 
&&\hspace{1cm}i=0,1,\dots,d-2, \quad
k=0,1,\dots,d-2.
\eeqn
This means that the setup is in fact $ d\otimes 2 \otimes d$. 
Applying local unitaries, $U_c=\sum_i \outerp{d-1-i}{i}$ on $C$ and $\uroj \sigma_y^B$ on $B$,  and then swapping $A$ and $B$ we get a tripartite subspace  $\calS_{2\times d^2}^{\pi/2} $ of the class considered in the main text in Theorem \ref{ges-ggm}. 

It is worth noting that the dimension of the GES displays the same behavior in the leading term as the maximal dimension, which here is $d^2-1$  \cite{upb-to-ges}.

\section{A class of $PPT$ states for which the value of the SDP bound is lower than  $E_{GM}(\calS_{2\times d^2}^{\pi/2})$}\label{exemplary-state}
The (unnormalized) spanning vectors  of  $\calS_{2\times 3^2}^{\pi/2}$ are: $\ket{000}+\ket{111}$, $\ket{001}+\ket{112}$, $\ket{010}+\ket{121}$, $\ket{011}+\ket{122}$. 
Let $\calP_{\calS}^{3,3}$ denote the projection onto the GES.

Consider the following state:
\begin{widetext}
\beqn
\rho=\frac{1}{N}
\left(
\begin{array}{cccccccccccccccccc}
	a \alpha ^2 & 0 & 0 & 0 & 0 & 0 & 0 & 0 & 0 & 0 & 0 & 0 & 0 & a \alpha  \beta  & 0 & 0 & 0 & 0 \\
	0 & \frac{b+c}{2} & 0 & 0 & 0 & 0 & 0 & 0 & 0 & 0 & 0 & 0 & 0 & 0 & \frac{b-c}{2} & 0 & 0 & 0 \\
	0 & 0 & x & 0 & 0 & 0 & 0 & 0 & 0 & 0 & 0 & 0 & 0 & 0 & 0 & 0 & 0 & 0 \\
	0 & 0 & 0 & \frac{b+c}{2} & 0 & 0 & 0 & 0 & 0 & 0 & 0 & 0 & 0 & 0 & 0 & 0 & \frac{b-c}{2} & 0 \\
	0 & 0 & 0 & 0 & a \beta ^2 & 0 & 0 & 0 & 0 & 0 & 0 & 0 & 0 & 0 & 0 & 0 & 0 & a \alpha  \beta  \\
	0 & 0 & 0 & 0 & 0 & z & 0 & 0 & 0 & 0 & 0 & 0 & 0 & 0 & 0 & 0 & 0 & 0 \\
	0 & 0 & 0 & 0 & 0 & 0 & x & 0 & 0 & 0 & 0 & 0 & 0 & 0 & 0 & 0 & 0 & 0 \\
	0 & 0 & 0 & 0 & 0 & 0 & 0 & z & 0 & 0 & 0 & 0 & 0 & 0 & 0 & 0 & 0 & 0 \\
	0 & 0 & 0 & 0 & 0 & 0 & 0 & 0 & y & 0 & 0 & 0 & 0 & 0 & 0 & 0 & 0 & 0 \\
	0 & 0 & 0 & 0 & 0 & 0 & 0 & 0 & 0 & y & 0 & 0 & 0 & 0 & 0 & 0 & 0 & 0 \\
	0 & 0 & 0 & 0 & 0 & 0 & 0 & 0 & 0 & 0 & z & 0 & 0 & 0 & 0 & 0 & 0 & 0 \\
	0 & 0 & 0 & 0 & 0 & 0 & 0 & 0 & 0 & 0 & 0 & x & 0 & 0 & 0 & 0 & 0 & 0 \\
	0 & 0 & 0 & 0 & 0 & 0 & 0 & 0 & 0 & 0 & 0 & 0 & z & 0 & 0 & 0 & 0 & 0 \\
	a \alpha  \beta  & 0 & 0 & 0 & 0 & 0 & 0 & 0 & 0 & 0 & 0 & 0 & 0 & a \beta ^2 & 0 & 0 & 0 & 0 \\
	0 & \frac{b-c}{2} & 0 & 0 & 0 & 0 & 0 & 0 & 0 & 0 & 0 & 0 & 0 & 0 & \frac{b+c}{2} & 0 & 0 & 0 \\
	0 & 0 & 0 & 0 & 0 & 0 & 0 & 0 & 0 & 0 & 0 & 0 & 0 & 0 & 0 & x & 0 & 0 \\
	0 & 0 & 0 & \frac{b-c}{2} & 0 & 0 & 0 & 0 & 0 & 0 & 0 & 0 & 0 & 0 & 0 & 0 & \frac{b+c}{2} & 0 \\
	0 & 0 & 0 & 0 & a \alpha  \beta  & 0 & 0 & 0 & 0 & 0 & 0 & 0 & 0 & 0 & 0 & 0 & 0 & a \alpha ^2 \\
\end{array}
\right), 
\eeqn
\end{widetext}
where the normalization factor is $N=2(a+b+c+2 x+y+2 z)$.
The quantity of interest is:
\beq\displaystyle \label{quantity}
v:=\tr \left(\calP_{\calS}^{3,3} \rho\right) = \frac{a-2a\alpha  \sqrt{1-\alpha ^2}+2c +4 x+2y+4 z}{2(a+b+c+2 x+y+2 z)}.
\eeq
One checks that for the following values of the parameters
\beqn
a=\frac{9}{10}, b=\frac{14}{25}, c=\frac{7}{25}, x=\frac{1}{30},y=\frac{1}{14},z=\frac{1}{7},\alpha=\frac{7}{25}, \non
\eeqn
it holds $\rho^{\Gamma_A} \ge 0$, $\rho^{\Gamma_B} \ge 0$, $\rho^{\Gamma_C} \ge 0$ and
 the value of $v$ is
\beq
v=\frac{239371}{568000}\simeq  0.4214 < 3/7.
\eeq
\section{A basis for $\calQ_1^{N,d}$ -- the GES from Theorem 2 of \cite{upb-to-ges}}\label{baza-th2}
Here we present a basis for the  GES considered in Theorem 2 of \cite{upb-to-ges}. In the current paper this subspace is named $\calQ_1^{N,d}$ (Section \ref{q1}).

By the definition, $\calQ_1^{N,d}$ is the subspace orthogonal to the set of the following product vectors ($\alpha \in \mathbb{C}$)
\beqn
&&\left(1,\alpha^{p_1},\dots, \alpha^{(d-1)p_1}\right)_{A_1}\otimes \left(1,\alpha,\alpha^2, \dots\alpha^{d^{N-1}-1}\right)_{\aaa{2}{k}} \non
&&\hspace{+0.5cm}= \left(\sum_{i=0}^{d-1} \alpha^{ip_1} \ket{i} _{A_1}\right) \otimes  \left(\sum_{k=0}^{d^{N-1}-1} \alpha^k \ket{k}_{\aaa{2}{k}} \right) \non
&&\hspace{+0.5cm}= \sum_{i=0}^{d-1}\sum_{k=0}^{d^{N-1}-1} \alpha^{ip_1+k} \ket{i}_{A_1}\ket{k}_{\aaa{2}{k}} 
\label{th2-nupb}
\eeqn
with
\beqn
p_1 &\equiv&\sum_{m=2}^Nd^{N-m}= \sum_{m=0}^{N-2} d^m.
\eeqn
Notice that it holds
\beq (d-1)p_1=d^{N-1}-1. 
\eeq
As shown in \cite{upb-to-ges}, the dimension of this GES is equal to $d^N-(2d^{N-1}-1)$.

From (\ref{th2-nupb}) we infer that we must identify possible realizations of $ip_1 +m$ for each value of this expression. Since there will be more than one realization for the latter we see that the vectors in the GES can be put into different groups within which they are not orthogonal, while the vectors from different groups are.

In what follows we omit the superscripts denoting parties -- it is assumed that the first ket corresponds to $A_1$ and the second one to $\aaa{2}{k}=A_2 \dots A_N $. For clarity of exposition we also omit the normalization factors $1/\sqrt{2}$ and denote for convenience:
\beqn
\quad p_i:=ip_1, \quad i=0, 1, \dots, d-1.
\eeqn

In the components of the vectors (\ref{th2-nupb}) there are no repeating monomials $\alpha^{i}$ with $i=0,1, \dots, p_1-1$. 
Consideration of the monomials $\alpha^{p_1+k}$ with $ k=0,1,\dots, p_1-1$ gives rise to the following $p_1$ vectors in the GES:
\beqn
&&\ket{0}\ket{p_1+k}-\ket{1}\ket{p_0+k},\\
&&\hspace{+2cm}   k=0,1,\dots, p_1-1. \nonumber
\eeqn
Monomials $\alpha^{p_2+k}$  give rise to $p_2=2 p_1$ vectors (we present them in a ''cyclic'' form):
\beqn
&&\ket{0}\ket{p_2+k}-\ket{1}\ket{p_1+k},\non
&&\ket{1}\ket{p_1+k}-\ket{2}\ket{p_0+k},\\
&&\hspace{+2cm}   k=0,1,\dots, p_1-1. \nonumber
\eeqn
The vectors are constructed in a similar manner for $\alpha^{p_m+k}$ and for $m=d-2$
 we have $p_{d-2}=(d-2)p_1$ vectors:
\beqn
&&\ket{0}\ket{p_{d-2}+k}-\ket{1}\ket{p_{d-3}+k},\non
&&\ket{1}\ket{p_{d-3}+k}-\ket{2}\ket{p_{d-4}+k},\\
&&\vdots\non
&&\ket{d-3}\ket{p_1+k}-\ket{d-2}\ket{p_0+k},\non
&&\hspace{+2.5cm}   k=0,1,\dots, p_1-1.\nonumber
\eeqn
Further, monomials $\alpha^{p_{d-1}}$ give rise to $d-1$ vectors:
\beqn
&&\ket{0}\ket{p_{d-1}}-\ket{1}\ket{p_{d-2}},\non
&&\ket{1}\ket{p_{d-2}}-\ket{2}\ket{p_{d-3}},\\
&&\vdots\non
&&\ket{d-2}\ket{p_1}-\ket{d-1}\ket{p_0}.\nonumber
\eeqn

Now, let us consider powers of $\alpha$ larger than  $\alpha_{p_{d-1}}$. 
Monomials $\alpha^{p_{d-1}+k+1}$ give rise to $(d-2)p_1$ vectors:
\beqn
&&\ket{1}\ket{p_{d-2}+k+1}-\ket{2}\ket{p_{d-3}+k+1},\non
&&\ket{2}\ket{p_{d-3}+k+1}-\ket{3}\ket{p_{d-4}+k+1},\\
&& \vdots\non
&& \ket{d-2}\ket{p_1+k+1}-\ket{d-1}\ket{p_0+k+1},\non
&&\hspace{+3.7cm}   k=0,1,\dots, p_1-1. \nonumber
\eeqn
%
Monomials $\alpha^{p_d+k+1}$ give rise to $(d-3)p_1$ vectors:
\beqn
&&\ket{2}\ket{p_{d-2}+k+1}-\ket{3}\ket{p_{d-3}+k+1},\non
&&\ket{3}\ket{p_{d-3}+k+1}-\ket{4}\ket{p_{d-4}+k+1},\\
&& \vdots\non
&&\ket{d-2}\ket{p_2+k+1}-\ket{d-1}\ket{p_1+k+1},\non
&&\hspace{+3.7cm}   k=0,1,\dots, p_1-1.\nonumber
\eeqn
This construction ends for monomials
 $\alpha^{(2d-4)p_1+k+1}=\alpha^{p_{d-2}+p_{d-2}+k+1}$ giving rise to $p_1$ vectors:
\beqn
&&\ket{d-2}\ket{p_{d-2}+k+1}-\ket{d-1}\ket{p_{d-3}+k+1},\\
&&\hspace{+3.7cm}   k=0,1,\dots, p_1-1. \nonumber
\eeqn

Succinctly, these (unnormalized) vectors may be written as follows  ($k=0,1,\dots,p_1-1$):
\beqn
&&\ket{i}\ket{p_{m-i}+k}-\ket{i+1}\ket{p_{m-i-1}+k}, \\
&& \quad m=1,2,\dots, d-2, i=0,1,\dots,m-1, \non
&&\ket{i}\ket{p_{d-i-1}}-\ket{i+1}\ket{p_{d-i-2}}, \\
 && \quad i=0,1,\dots,d-2,\non
&&\ket{i}\ket{p_{m-i}+k+1}-\ket{i+1}\ket{p_{m-i-1}+k+1}, \\
&& \quad m=d-1,d,\dots, 2(d-2), i=m-(d-2),\dots, d-2.\nonumber
\eeqn
%

%
%

In the cyclic form that we have used,  two neighboring vectors have an overlap of $1/2$, while the rest are orthogonal within each group. This  observation is useful for the Gram--Schmidt procedure and we have the following.

\begin{lem}
Let there be given a subspace $\calS$ spanned by the vectors:
\beqn
\ket{\psi_k}=\frac{1}{\sqrt{2}}\left(\ket{\gamma_{k-1}}-\ket{\gamma_k}\right),\quad k=1,\dots,S,
\eeqn
with orthonormal $\ket{\gamma_i}$'s.

The following vectors are an orthonormal basis for $\calS$:
\beqn \label{orto-baza}
\hspace{-0.5cm}
\ket{\varphi_m}&=&\frac{1}{\sqrt{m(m+1)}}\left(\sum_{i=0}^{m-1}\ket{\gamma_i}-m\ket{\gamma_m}\right)\\
&=& \frac{1}{\sqrt{m(m+1)}}\sum_{i=0}^{m}\lambda_i\ket{\gamma_i}, \;\; m=1,\dots, S,
\eeqn
where $\lambda_0=\lambda_1=\dots=\lambda_{m-1}=1$ and $\lambda_m=-m$.
\end{lem}
\begin{proof}
We present the proof assuming the unnormalized case for simplicity.
Set 
\beq\ket{\varphi_1}=\ket{\psi_1}.\eeq
 The remaining vectors are formed  as follows
 \beqn
 &&\ket{\varphi_2}=\ket{\varphi_1}+2\ket{\psi_2},\non
&&\ket{\varphi_3}=\ket{\varphi_2}+3\ket{\psi_3},\\
&&\vdots \nonumber
\eeqn
that is
\beqn
\ket{\varphi_k}=\ket{\varphi_{k-1}}+k\ket{\psi_k}.
 \eeqn
 It is easy to see that the above procedure gives an orthonormal basis (\ref{orto-baza}).
\end{proof}

Concluding,  we observe that:
\beqn
\ket{p_i+k}_{\aaa{2}{N}}&=& \ket{i+k_{N-2}}_{A_2}\cdots\ket{i+k_1}_{A_{N-1}}\ket{i+k_0}_{A_N}\non
 &&\hspace{+0.5cm}= \bigotimes_{m=2}^N \ket{i+k_{N-m}}_{A_m}
\eeqn
with $k=\sum_{l=0}^{N-2}k_l d^l$ 
and the constraint that 
$k\le p_1-1$.

\section{A basis for $\calQ_2^{N,d}$ -- the GES from Theorem 2 of \cite{upb-to-ges}}\label{baza-th3}

In this appendix we consider the GES from  Theorem 3 of \cite{upb-to-ges} and present a non--orthogonal basis for it. In the main text this GES is named  $\calQ_2^{N,d}$ (Section \ref{q2}).  A basis for the  qubit case has been given in the cited paper as an example and we now solve for the general case. In fact, the basis we obtain is a simple generalization of Eq. (64) from \cite{upb-to-ges}.

The GES under inspection has been defined as the subspace orthogonal to the set of the following product vectors ($\alpha \in \mathbb{C}$):
\beqn
&&\left(1,P_1(\alpha),\dots, P_{d-1}(\alpha)\right)_{A_1}\otimes \left(1,\alpha,\alpha^2, \dots\alpha^{d^{N-1}-1}\right)_{\aaa{2}{k}}, \non
&&\hspace{+0.5cm} = \left(\sum_{k=0}^{d-1} P_k(\alpha) \ket{k} _{A_1} \right)\otimes  \left(\sum_{m=0}^{d^{N-1}-1} \alpha^m \ket{m}_{\aaa{2}{k}}  \right)\non
&&\hspace{+0.5cm} = \sum_{k=0}^{d-1}\sum_{m=0}^{d^{N-1}-1} P_k(\alpha)\alpha^{m} \ket{k}_{A_1}\ket{m})_{\aaa{2}{k}} ,
\label{th3-nupb}
\eeqn
where
\beqn
P_k(\alpha)=\sum_{m=2}^N \alpha^{kd^{N-m}}.
\eeqn

It has been observed in \cite{upb-to-ges} that only the polynomials $\alpha^m P_k(\alpha)$ with $m=0,1,\dots, d^{N-1}-d^{N-2}-1$ are linear combinations of other components of  (\ref{th3-nupb}), more precisely they are sums of the monomials which are the first $d^{N-1}-1$ components of these vectors. The number of these linearly dependent polynomials  gives us thus the dimension of the GES and we easily construct the vectors spanning the GES using the above observations:
\beqn
&&\ket{0}_{A_1}\left( \sum_{f=2}^{N} \ket{kd^{N-f}+m}_{\aaa{2}{k}}   \right)  - \ket{k}_{A_1}\ket{m}_{\aaa{2}{k}}, \\
&& \hspace{+0.5cm} k=1,\dots, d-1, \quad m=0,1,\dots, d^{N-1}-d^{N-2}-1.\nonumber
\eeqn

Clearly, these vectors do not form an orthogonal set, which is the main obstacle in analytical computation of the measure for this subspace. That is, after the orthogonalization their form is very involved and seem not to offer an easy insight into the structure of the GES.

\end{document}